\documentclass[a4paper,11pt]{article}
\usepackage{amsmath, amssymb}
\usepackage{subfig}
\usepackage{graphicx}
\usepackage{color}
\usepackage{graphicx}
\usepackage{url}
\usepackage{color,soul}
\usepackage{float}
\usepackage{amsmath}
\usepackage{amssymb}
\usepackage{courier}
\usepackage{multicol}
\usepackage[abs]{overpic}
\usepackage{enumitem}
\usepackage{multicol}
\usepackage{amsthm}
\usepackage{changepage}
\usepackage{framed}
\usepackage{mdframed}
\usepackage{lipsum}
\usepackage{longtable}
\usepackage{tablefootnote}
\usepackage{threeparttable}
\newmdenv[
linecolor=gray,
linewidth=3pt,
topline=false,
bottomline=false,
rightline=false,
rightmargin=-10pt,
leftmargin=0pt,
skipabove=\topsep,
skipbelow=\topsep
]{siderules}

\newtheorem{theorem}{Theorem}
\newtheorem{lemma}[theorem]{Lemma}

\usepackage[linesnumbered,boxed]{algorithm2e}
\allowdisplaybreaks
\usepackage{pgf,tikz}
\usepackage{mathrsfs}
\usetikzlibrary{arrows}
\usepackage{times}
\DeclareMathSizes{10}{9}{6}{5}

\def\cY{{\mathcal Y}}

\def\Tr{\text{Tr}}
\newcommand{\defeq}{\stackrel{\mathrm{def}}{=}}

\def\BB{\mathbb B}

\def\RR{\mathbb R}
\def\QQ{\mathbb Q}
\def\ZZ{\mathbb Z}

\def\CC{\mathbb C}

\def\SS{\mathbb S}

\def\cA{\mathcal A}
\def\cB{\mathcal B}
\def\cC{\mathcal C}

\def\cF{\mathcal F}
\def\cG{\mathcal G}
\def\cH{\mathcal H}
\def\cI{\mathcal I}

\def\cS{\mathcal S}
\def\cP{\mathcal P}
\def\cT{\mathcal T}

\def\cL{\mathcal L}

\def\cH{\mathcal H}
\def\cE{\mathcal E}
\def\cV{\mathcal V}

\def\cX{\mathcal X}

\def\bx{\mathbf x}

\def\bT{\mathbf T}

\def\bI{\mathbf I}

\def\by{\mathbf y}
\def\bu{\mathbf u}
\def\ba{\mathbf a}

\def\bc{\mathbf c}
\def\bb{\mathbf b}

\def\bz{\mathbf z}
\def\bv{\mathbf v}
\def\bw{\mathbf w}
\def\bzero{\mathbf 0}

\def\b1{\mathbf 1}

\def\cI{{\mathcal I}}
\def\cY{{\mathcal Y}}

\newcommand{\vol}{\operatorname{vol}}
\newcommand{\poly}{\operatorname{poly}}

\newcommand{\polylog}{\operatorname{polylog}}
\newcommand{\rank}{\operatorname{rank}}

\newcommand{\Diag}{\operatorname{Diag}}

\newtheorem{proposition}{Proposition}
\newtheorem{fact}{Fact}

\newtheorem{remark}{Remark}

\newtheorem{corollary}{Corollary}
\newcommand{\raf}[1]{(\ref{#1})}

\usepackage[top=2.5cm, bottom=2.5cm, left=2.5cm, right=2.5cm]{geometry}

\newcommand{\quasipoly}{\operatorname{quasi-poly}}

\newcommand{\hide}[1]{}

\newcommand{\bone}{\ensuremath{\boldsymbol{1}}}
\newcommand{\add}[1]{{\color{black}{#1}}}

\title{Dual Bounded Generation: Polynomial, Second-order Cone and Positive Semidefinite Matrix Inequalities}
\author{ 
Khaled Elbassioni\thanks {Khalifa University of Science and Technology, Abu Dhabi, UAE;
(khaled.elbassioni@ku.ac.ae)}
}

\begin{document}
\date{}
\maketitle
\begin{abstract}
%The well-known monotone Boolean dualization (MBD) problem calls for finding the irredundant CNF corresponding to a given monotone DNF. It had been shown that the problem is (quasi-)polynomially equivalent to that of enumerating all maximal feasible solutions of a given monotone system of linear/separable/supermodular inequalities over binary vectors. Continuing in this line of work, we consider systems of polynomial, second-order cone, and semidefinite inequalities. We give sufficient conditions under which the enumeration problem remains equivalent to MBD and highlight some applications. In particular, we show that all maximal feasible solutions for a given chance-constrained knapsack inequality can be enumerated in quasi-polynomial time when the underlying covariance matrix is of constant cp-rank.  Whether or not these sufficient conditions are also necessary remains an interesting open question.  
In the monotone integer dualization  problem, we are given two sets of vectors in an integer box such that no vector in the first set is dominated by a vector in the second.  The question is  to check if the two sets of vectors cover the entire integer box by upward and downward domination, respectively. It is known that the problem is (quasi-)polynomially equivalent to that of enumerating all maximal feasible solutions of a given monotone system of linear/separable/supermodular inequalities over integer vectors. The equivalence is established via showing that the dual family of minimal infeasible vectors has size bounded by a (quasi-)polynomial in the sizes of the family to be generated and the input description.  Continuing in this line of work, in this paper, we consider systems of polynomial, second-order cone, and semidefinite inequalities. We give sufficient conditions under which such bounds can be established and highlight some applications. 
\end{abstract}

\section{Introduction}
\label{intro}

We consider a {\it monotone} system of inequalities of the form:
\begin{alignat}{3}
\label{mon}
\quad & \displaystyle f_i(\bx)\leq ~t_i,\quad\text{ for }i\in[r]:=\{1,\ldots,r\},
\end{alignat}
over a vector of {\it integer} variables $\bx=(x_1,\ldots,x_n)\in\ZZ_+^n$, where $f_i:\ZZ_+^n\mapsto\RR_+$ is a {\it monotone} (non-decreasing) non-negative function on $\ZZ_+^n$, that is, $\bx,\by\in \ZZ_+^n$ and $\bx\geq \by$ imply that $f_i(\bx)\ge f_i(\by)$, for all $i\in[r]$. A vector $\bx\in\ZZ_+^n$ is said to be a {\it maximal feasible vector} (or solution) for \raf{mon} if $\bx$ is feasible for \raf{mon} and $\bx+\b1^j$ is not feasible for all $j\in[n]$, where throughout we use $\b1^j$ to denote the $j$-{th} unit $n$-dimensional vector. Likewise, a vector $\bx\in\ZZ_+^n$ is said to be a {\it minimal infeasible vector} for \raf{mon} if $\bx$ is infeasible for \raf{mon} and $\bx-\b1^j$ is feasible for all $j\in[n]$ such that $x_j>0$.
Let $\cF$ and $\cI$ be respectively the families of maximal feasible and minimal infeasible vectors for \raf{mon}. We are interested in {\it incrementally} generating the family $\cF$:

\begin{description}
	\item[{\em GEN}$(\cY)$:] {\em Given a monotone 
		system~\raf{mon}, and a subfamily $\cY \subseteq \cF$ of its
		maximal feasible vectors, either find a new maximal
		vector $\bx\in\cF\setminus\cY$, or \add{confirm} that 
		$\cY=\cF$.}
\end{description}

\noindent Clearly, the entire family $\cF$ can be generated by initializing $\cY=\emptyset$ and iteratively solving the above
problem $|\cF|+1$ times.  It was shown in~\cite{BEGKM02-SICOMP} that, when each $f_i$ is a {\it linear} function, problem {\em GEN}$(\cY)$ can be solved in {\it quasi-polynomial} time $k^{o(\log k)}$ time, where $k=\max\{n, r, |\cY|\}$, while the similar incremental generation problem for the family of minimal infeasible solutions is NP-hard. This result was extended to  the case when each function $f_i$ can be written as the sum $f_i(\bx)=\sum_{j=1}^nf_{ij}(x_j)$ of {\it single-variable} monotone functions $f_{ij}$, and more generally to the case when each $f_i$ is the sum of products of {\it constant number} of single-variable monotone functions. A particularly interesting example of the latter case is when each $f_i$ is a polynomial:
\begin{align}\label{poly}
f_i(\bx)=\sum_{H\in\cH_i}a_H\prod_{j\in H} x_j^{d_{H,j}}, 
\end{align}
where each $\cH_i\subseteq 2^{[n]}$ is a given {\it mutliset} family with %$|H|=O(1)$, 
$d_{H,j}\in\ZZ_+\setminus\{0\}$ and $a_H>0$ for all $H\in\cH_i$.

\medskip
It will be convenient\add{\footnote{\add{It is easy to see that the family $\cF$ must be finite and each element in $\cF$ is bounded. Indeed, suppose that there is an infinite sequence of elements $\bx(\ell)\in\cF$, $\ell=1,2,\ldots$. Then there must exist a $j\in[n]$ such that $x_j(\ell)\rightarrow\infty$ as $\ell\rightarrow\infty$. But then, for any $i\in[r]$, we would have $t_i\ge\lim_{\ell\rightarrow\infty} f_i(\bx(\ell))=\lim_{\ell\rightarrow\infty} f(\bx(\ell)+\b1^j)$, contradicting the fact that $\bx(\ell)\in\cF$. As an example, consider the inequality $x_1x_2\le  1$ over $\ZZ_+^2$. Then $\cF=\{(1,1)^\top\}$, while $(0,z)^\top$ and $(z,0)^\top$ are not in $ \cF$ for all $z\in\ZZ_+$.}}} to restrict the domain of each variable $x_j$ to a subset $\cC_j=\{0,1\ldots,c_j\}$ of $\ZZ_+$. Such a bound $c_j$ of {\it polynomial bit-length} can be assumed in all the examples considered in this paper. For instance, if each function $f_i$ is a polynomial of the form~\raf{poly}, then any feasible solution for~\raf{mon} satisfies\footnote{assuming that $x_j\ne\infty$.} $$x_j\le c_j':=\add{\min_i\frac{ t_i}{\min_{H\in\cH_i:~j\in H} a_{H}}}.$$ Thus, we may set $c_j:=\lfloor c_j'\rfloor$, for all $j\in[n]$. Keeping this in mind, we will assume in the rest of the paper that the variable vector $\bx$ is chosen from an integer box
%\footnote{In fact, it is not difficult to see that we may more generally assume each $\cC_j$ to be any arbitrary finite subset of $\RR_+$ without affecting the bounds obtained in the paper.}
 \add{$\cC=\cC_1\times \cC_2\times \cdots \times \cC_n=\{\bx\in\ZZ^n|~\bzero\le\bx\le \bc\}$}, where $\cC_j=\{0,1\ldots,c_j\}$ for $j\in [n]$.

\medskip

For an {\it antichain} (that is, a subset of pairwise incomparable elements) $\cA\subseteq\cC$, denote by $\cI(\cA)$ the set of {\it minimal non-dominated} elements of $\cA$, i.e., the set of those
elements $\bx\in\cC$ that are minimal with respect to the property
that $\bx\not\leq \ba$ for all $\ba\in\cA$. It is easy to see that the mapping $\cI:\text{Antichains}(\cC)\to\text{Antichains}(\cC)$ is one-to-one, and hence, the families $\cA$ and $\cI(\cA)$ can be thought of {\it dual} to each other. In particular, if $\cF$ is the family of maximal feasible solutions for~\raf{mon}, then the dual 
$\cI(\cF)$ represents the family of {\em minimal} vectors
of $\cC$ which {do not} satisfy \raf{mon}.

Let $f : \cC \mapsto\RR_+$, be a real-valued function over $\cC$. $f$ is said to be {\it supermodular} if
\begin{align}\label{sup}
f(\bx \vee \by) + f(\bx \wedge \by) \geq  f(\bx) + f(\by)
\end{align}
holds for all $\bx,\by\in\cC$, where $\vee$ and $\wedge$ denote, respectively, the component-wise maximum and minimum operators over $\cC$: $(\bx \vee \by)_j=\max\{x_j,y_j\}$ and $(\bx \wedge \by)_j=\min\{x_j,y_j\}$, for $j\in[n]$. 
 It was shown in \cite{BEGK-DAM03} that if all the functions $f_i$ in~\raf{mon} are {\it integer}-valued supermodular functions, then {\em GEN}$(\cY)$ can be solved in quasi-polynomial time $k^{o(\log k\cdot\log R)}$ time, where $k=\max\{n, r,|\cY|\}$ and $R=\max_{i\in[r]}f_i(\bc)$. In particular, if all functions have {\it quasi-polynomially bounded integral range} then all maximal feasible vectors for the system can be enumerated in quasi-polynomial time. This, as well as all the above-mentioned results, were established via a (quasi-)polynomial time reduction to the following {\em dualization} problem on
integer boxes: %(see \cite{BEGKM02-SICOMP} and also \cite{BI95,GK99}):

\begin{description}
	\item[{\em DUAL}$(\cC, \cA,\cB)$:] {\em Given an integer box $\cC$, an
		antichain of vectors $\cA\subseteq \cC$ and a subset $\cB \subseteq
		\cI(\cA)$ of its minimal non-dominated vectors, either find a new
		minimal non-dominated vector $\bx \in \cI(\cA)\setminus\cB$, or \add{confirm} that no such vector exists, i.e., $\cB=\cI(\cA)$.}
\end{description}
It is known that problem {\em DUAL}$(\cC,\cA,\cB)$ can be
solved in $poly(n)$ $+$ $m^{o( \log m)}$ time, where
$m=|\cA|+|\cB|$ (see \cite{BEGKM02-SICOMP,FK96}).
However, it is still open whether {\em DUAL}$(\cC,\cA,\cB)$ has a polynomial time algorithm.
% (e.g., \cite{BI95,FK96,MI97}). 
To recall (at a high level) the reduction from problem {\em GEN}$(\cY)$ to problem {\em DUAL}$(\cC,\cA,\cB)$, let us follow~\cite{BGKM01} in calling the family $\cF$ {\em uniformly
	dual-bounded}, if for every subfamily $\cY\subseteq \cF$ we have
\begin{equation}\label{udbdd}
|\cI(\cY)\cap \cI(\cF)|\leq q(\cL,|\cY|)
\end{equation}
for some (quasi-)polynomial $q(\cdot)$, where $\cL$ is the total length of the binary encoding of the parameters defining the system~\raf{mon}. It was shown in~\cite{BEGKM02-SICOMP} (extending the similar result  in~\cite{BI95,GK99}, for \add{the} binary case) that the incremental {\it joint} generation problem of enumerating (in a {\it non-controlled} order) the union family $\cF\cup\cI(\cF)$ can be reduced in polynomial time to the dualization problem {\em DUAL}$(\cC, \cdot,\cdot)$ over an integer box $\cC$. It follows then that for uniformly dual-bounded families $\cF$, problem GEN$(\cY)$ can be reduced in (quasi-)polynomial time to the dualization problem by performing joint generation and discarding the elements of the unwanted family $\cI(\cY)\cap\cI(\cF)$. Since the number of discarded elements is bounded by a (quasi-)polynomial $q(\cL,|\cY|)$ in the input-output size, the overall running time for enumerating $\cF$ is dominated by a (quasi-)polynomial in $q(\cL,|\cY|)$ and the time $T_\text{\em DUAL}$ needed to solve the dualization problem, resulting in an overall quasi-polynomial time, if the family $\cF$ is uniformly dual-bounded. 

\medskip

We now state some of the known bounds of the form~\raf{udbdd} for families $\cF$ of maximal feasible vectors for \raf{mon}:
\begin{itemize}
	\item [(B1)] When the all the functions $f_i$ are integer-valued supermodular %\footnote{Such functions are called {\it polymatriod} provided that $f_i(\bzero)=0$.} 
	of range $\{0,1,\ldots,R\}$, we have $q(\cL,|\cY|)\le r|\cY|^{o(\log(R-t))}$, where $t:=\min\{t_1,\ldots,t_r\}$ \cite{BEGK-DAM03}. 
	\item[(B2)] If each $f_i$ is linear, that is, $f_i(x)=(\ba^i)^\top\bx$, for some non-negative vector $\ba^i$, then $q(\cL,|\cY|)\le rn|\cY|$ \cite{BEGKM02-SICOMP}. In fact, this is also true for the case when each $f_i(\bx)=\sum_{j=1}^nf_{ij}(x_j)$ is the sum of {\it single-variable} monotone functions \cite{KBEGM07}.
	\item[(B3)] More generally, if each $f_i$ is the sum of at most $s$ terms each of which is the product of at most $p$ single-variable monotone functions, then $q(\cL,|\cY|)\le rsp(2|\cY|+1)^p$ \cite{KBEGM07}. This is true, in  particular, for a system of polynomial inequalities of the form~\raf{poly}, when the {\it dimension} of each hypergraph $\cH_i$ (that is, the maximum size of a hyperedge) is bounded by a constant.
	
\end{itemize}
   
 \iffalse
The general framework that was used in \cite{BEGM02} to prove the above dual-bounding inequalities is to consider two finite sets $\cX$ and $\cY$ of points in $\cC$ such that
\begin{itemize}
	\item[(P1)] for any two distinct points $\bx,\bx'\in \cX$, their component-wise minimum	$\bx \wedge\bx'$ is dominated by some  $\by \in \cY$, i.e., $\bx\wedge\bx'\leq\by$; and
	
	\item[(P2)] $\cX$ and $\cY$ can be separated by a nonnegative function belonging to a certain class (e.g., linear supermodular, etc):  $f(\bx) > t \geq f(\by)$ for all $\bx \in \cX$ and
	$\by \in \cY$, where $f$ is a function  $t\in \RR$ is a given real threshold.
\end{itemize}

Given $\cX,\cY\subseteq\cC$ satisfying properties (P1) and
(P2), one may ask the question of how large the size of $\cX$
can be in terms of the size of $\cY$. For instance, if $\cX$ is
the set of the $n$-dimensional unit vectors, $\cY=\{\bzero\}$ 
is the set containing only the origin, and $f(\bx)=\sum_ix_i$, then $\cX$ and $\cY$
satisfy properties (P1), (P2), and the ratio between their
cardinalities is $n$. It was shown in\cite{KBEGM07} that this is actually  an
extremal case:

\begin{lemma}[\cite{KBEGM07}]\label{intersL}
	If $\cX$ and $\cY\neq \emptyset$ are two finite sets of points in $\RR^n$ satisfying properties {\rm (P1)} and {\rm (P2)} above,
	then
	\begin{equation*}\label{ineq}
	|\cX|\leq n|\cY|.
	\end{equation*}
\end{lemma}	
\fi

\paragraph{Main results.} We extend the above results (B1)-(B3) as follows:

\begin{itemize}
	\item[(B4)] We consider {\it real}-valued supermodular functions with range $[0,R]$ and obtain a bound $q(\cL,|\cY|)\le r|\cY|^{o(\log\frac{R-t}{\tau})}$ that depends on the minimum {\it traction} (i.e., minimum possible positive change \add{in the value of the fucntion that can  result from increasing the argument only along one coordinate\footnote{\add{Note that the minimum possible positive change $\kappa(f)$ in the value of a (supermodular) function $f$ can be much smaller than the traction $\tau(f)$ (as we define it). For example, the (monotone linear) function  $f:\{0,1\}^2\to\RR_+$ defined as $f(\bx):=(1+\epsilon)x_1+x_2$, where $\epsilon>0$ is an arbitrarily small constant, has $\kappa(f)=\epsilon$, while $\tau(f)=1$.}  }}) $\tau$ of the functions $f_i$. In particular, when $\frac{R-t}\tau=\quasipoly(\cL,|\cY|)$, we obtain a quasi-polynomial time enumeration algorithm.
	
	\item[(B5)] As a direct application of the result in (B4), we consider the case when each function $f_i$ is a product of real-valued {\it affine} functions with rational coefficients and the objective is to enumerate the family of all {\it minimal feasible} solutions of the system $f_i(\bx)\ge t_i$ for $i=1,\ldots,r$. We derive a bound of $q(\cL,|\cY|) \le|\cY|^{o(\log\cL)} $ on the size of the dual family, implying that the problem  of enumerating all minimal feasible solutions of such systems can be solved in quasi-polynomial time. 
	   
	\item[(B6)] We show that, if each $f_i$ is an integer-valued polynomial function (as in \raf{poly}) of range $\{0,1,\ldots,R\}$ having at most $s$ terms in which each variable has degree at most $d$, then \add{$q(\cL,|\cY|)\le r|\cY| (R-\min_it_i,$ $s+2^{d+2}n\max_it_i|\cY|)$}. In particular, if $d=\polylog(\cL)$ and, for all $i$, either $t_i\ge R_i-\quasipoly(\cL)$  or $t_i\le\quasipoly(\cL)$, then all maximal feasible vectors for~\raf{mon} can be enumerated in quasi-polynomial time. In contrast to the result in (B3), this does not require the hypergraphs defining the polynomials in~\raf{poly} to have fixed dimension.
	\item[(B7)] We consider the case when~\raf{mon} is a system of {\it second-order cone} inequalities, that is, when each $f_i$ is a real-valued function of the form $f_i(\bx):=\|A^i\bx\|+(\bb^i)^\top\bx$, where $A^i\in\RR_+^{d\times n}$ and $\bb^i\in\RR_+^n$ are given matrices and vectors, and $\|\cdot\|$ denotes the $\ell_2$-norm. We show in this case that $q(\cL,|\cY|)\le O(n)^{2d+1}r|\cY|$. In particular if $d=\polylog(\cL)$ then problem {\em GEN}$(\cY)$ can be solved in quasi-polynomial time.
	\item[(B8)] Finally, we consider the case when each  $f_i$ is a real-valued function of the form $f_i=\lambda_{\max}(\sum_{j=1}^nA^{i,j}x_j)$ where
 $A^{i,j}\in\RR_+^{d\times d}$ is a {\it positive semidefinite} matrix, and $\lambda_{\max}(X)$ denotes the maximum eigenvalue of the matrix $X$. This gives rise to a semidefinite inequality system. We show in this case that $q(\cL,|\cY|)\le O(n)^{2d+1}r|\cY|$, and consequently, if $d=\polylog(\cL)$ then problem {\em GEN}$(\cY)$ can be solved in quasi-polynomial time.
\end{itemize}
	 
We give some motivating applications of these results in the next section. Proofs of the bounds claimed in (B4), (B5), (B6), (B7) and (B8) are given in Sections~\ref{sec:supermodular}, ~\ref{sec:prodcut},~\ref{sec:polynomial}, \ref{sec:soc} and~\ref{sec:psd}, respectively. To simply our presentation, it will be enough to establish the bound for a single inequality:
\begin{alignat}{3}
\label{mon-s}
\quad & \displaystyle f(\bx)\leq ~t,~\quad \bx\in\cC,
\end{alignat}
where we assume w.l.o.g. that $0\le t\le f(\bc)$.
Indeed, given a system~\raf{mon}, let us denote by $\cF_i$ the set of maximal feasible solutions of the inequality $f_i(\bx)\le t_i$. Suppose that we manage to show that $|\cI(\cY)\cap\cI(\cF_i)|\le q_i(\cL,|\cY|)$ for any $\cY\subseteq\cF_i$. Then, for any subset $\cY\subseteq\cF$ of maximal feasible solutions of the system, a union bound can be applied to obtain\add{\footnote{\add{To see that $\cX:=\cI(\cY) \cap  \cI(\cF)=\cY:=\bigcup_{i=1}^r\left(\cI(\cY) \cap  \cI(\cF_i)\right)$, note first that $\cX\subseteq\cY$ since $\cI(\cF)\subseteq\bigcup_{i=1}^r \cI(\cF_i)$. On the other hand, for any $\bx\in\cY$, $\bx\in \cI(\cY) \cap  \cI(\cF_i)$ for some $i\in[r]$. Then, $\bx\in\cI(\cF_i)$ implies that $f_i(\bx)>t_i$, that is, $\bx$ is infeasible for the system, while $\bx\in\cI(\cY)$ implies that $\bx-\bone^j\le \by$, for some $\by\in\cY\subseteq\cF$, whenever $x_j>0$, that is, $\bx-\bone^j$  is feasible for the system. It follows that $\bx\in \cX$.}}}
\begin{align*}
|\cI(\cY) \cap  \cI(\cF)|=\left|\bigcup_{i=1}^r\left(\cI(\cY) \cap  \cI(\cF_i)\right)\right|\le\sum_{i=1}^r|\cI(\cY) \cap  \cI(\cF_i)|\le\sum_{i=1}^rq_i(\cL,|\cY|).
\end{align*}
Thus, in deriving the stated bounds, we will focus on a single inequality in the system~\raf{mon}.
\section{Some Applications}

\paragraph{Fair allocation of goods.}
Consider a set of $m$ agents and $n$ types of indivisible goods (or items). We assume there is {\it unlimited} supply of each item.  Each agent $i$ demands at least $t_i$ items to be allocated to it and its utility for receiving an allocation $\bx^i:=(x_{ij}~|~j\in[n])$ is given by the linear function $u_i(\bx):=\sum_ja_{ij}x_{ij}$, where $x_{ij}\in\ZZ_+$ is the number of goods of type $j$ allocated to agent $i$. Given a threshold parameter $t$, we are interested in finding all minimal allocations satisfying the demand constraints and achieving a {\it Nash social welfare} of value at least $t$:
\begin{alignat}{3}
	\label{nsw}
	\quad & \displaystyle \left(\prod_{i=1}^mu_i(\bx)\right)^{1/m}\ge t,\\	
	\quad & \sum_{j}x_{ij}\ge t_i,\quad\text{ for }i\in[m],\\
	&\bx\in\ZZ_+^{mn}.\nonumber
\end{alignat}
 This gives rise to a monotone system of inequalities, each of which is involving either a linear function  or a product of linear functions. It follows from the result in (B5) that the family of minimal feasible solutions for this system is uniformly dual-bounded and hence can be enumerated in incremental quasi-polynomial time.  

\paragraph{Chance-constrained \add{multi}-dimensional knapsack inequalities.}
Consider a system~\raf{mon} where each $f_i(\bx):=(\bw^i)^\top\bx$ is a linear function defined by a non-negative weight vector $\bw^i\in\RR^n$. The elements of $[n]$ can be interpreted as items to be packed into $r$ knapsacks of capacities $t_1,\ldots,t_r$, where $w^i_{j}$ represents the size requirement of item $j$ in knapsack $i$. In the stochastic version, each vector $\bw^i$ is drawn from a {\it multivariate normal} distribution with mean $\ba^i \in\RR_+^n$ and covariance matrix $\Sigma^i\in\SS^{n\times n}_+$, i.e., $\bw^i\sim N(\ba^i,\Sigma^i)$. The requirement is to pack the items into the the knapsacks, such that the $i$th capacity constraint is satisfied with probability at  least $\alpha_i\in[0,1]$: 
\begin{alignat}{3}
\label{kp}
\quad & \displaystyle \Pr[(\bw^i)^\top\bx\leq ~t_i]\ge\alpha_i,\quad\text{ for }i\in[r],\\
&\bx\in\{0,1\}^n.\nonumber
\end{alignat}

As $(\bw^i)^\top\bx\sim N((\ba^i)^\top\bx,\bx^\top\Sigma^i\bx)$, we can reformulate the constraints~\raf{kp} as: 
\begin{alignat}{3}
\label{kp2}
\quad & \displaystyle (\ba^i)^\top\bx+\Phi^{-1}(\alpha_i)\sqrt{\bx^\top\Sigma^i\bx}\le t_i,\quad\text{ for }i\in[r]:=\{1,\ldots,r\},\\
&\bx\in\{0,1\}^n,\nonumber
\end{alignat} 
where $\Phi(\cdot)$ represents the cumulative distribution function of the standard normal distribution. While the complexity of enumerating the family of maximal feasible solutions for~\raf{kp2}, in general, remains open at this point, we can efficiently solve the problem in two special cases described below. In both cases, we assume that $\alpha_i\ge 0.5$ and hence $\Phi^{-1}(\alpha_i)\ge0$.

\begin{itemize}
	\item Fixed-rank case: this is the case when the  covariance matrices $\Sigma^i$ have {\it completely positive} (cp) rank $d_i$, i.e., we can find matrices $A^i\in\RR_+^{d_i\times n}$ such that  $\Sigma^i=(A^i)^\top A^i$ (this is, for example, the case when $\bw^i=(A^i)^\top \bz+\ba^i$, where $z_1,\ldots, z_{d_i}\sim N(0,1)$ are i.i.d.'s.). In such a case, we can rewrite~\raf{kp2} as a second order cone program:
	\begin{alignat}{3}
		\label{kp3}
		\quad & \displaystyle (\ba^i)^\top\bx+\Phi^{-1}(\alpha_i)\|A^i\bx\|\le t_i,\quad\text{ for }i\in[r]:=\{1,\ldots,r\},\\
		&\bx\in\{0,1\}^n.\nonumber
	\end{alignat} 
	It follows then from the result in (B7) that, when $\max_id_i=O(1)$, all maximal feasible solutions of \raf{kp} can be enumerated in quasi-polynomial time.
	
	\item Ordered independent case: this is the case when item sizes are independent: $A^i=D^i:=\Diag(d^i_1,\ldots,d^i_n)$ is a full-rank {\it diagonal} matrix, and in addition, we are given permutations $\sigma_1,\ldots,\sigma_r:[n]\to[n]$ s.t. for each $i\in[r]$:
	\begin{align}\label{order}
	 a^i_{\sigma_i(1)}\ge\cdots\ge a^i_{\sigma_i(n)} \quad\text{ and }\quad d^i_{\sigma_i(1)}\ge\cdots \ge d^i_{\sigma_i(n)},
	 \end{align}	
 that is, in each knapsack, the means and standard deviations of the item sizes can be ordered in the same way 
(this is true, for example, when $w_j^i\sim N(a^i_j,1)$ are independent). In this case, we can verify that the function $f_i(\bx):=(\ba^i)^\top\bx+\Phi^{-1}(\alpha_i)\|A^i\bx\|=(\ba^i)^\top\bx+\Phi^{-1}(\alpha_i)\sqrt{\sum_j(d_j^i)^2x_j}$ is {\it 2-monotonic} with permutation $\sigma_i$; see the definition in Section~\ref{sec:soc}. Indeed, for any $\bx\in\cC$ and $k,j\in[n]$ with $k<j$, $\bx_{\sigma_i(k)}=0$ and $\bx_{\sigma_i(j)}=1$, it holds by~\raf{order} that 
\begin{align*}
	\quad &f_i(\bx+\b1^{\sigma_i(k)}-\b1^{\sigma_i(j)})-f(\bx)=\\
	\quad & \quad\qquad a^i_{\sigma_i(k)}-a^{i}_{\sigma_i(j)}+\frac{\Phi^{-1}(\alpha_i)\left[(d^i_{\sigma_i(k)})^2-(d^i_{\sigma_i(j)})^2\right]}{\add{\sqrt{\sum_{j'}(d^i_{j'})^2x_{j'}+(d^i_{\sigma_i(k)})^2-(d^i_{\sigma_i(j)})^2}}+\add{\sqrt{\sum_{j'}(d^i_{j'})^2x_{j'}}}}
			\ge0.
	\end{align*} 
	Thus, we can use Lemma~\ref{l2} below to derive the the bound $|\cI(\cY)\cap \cI(\cF)|\le rn|\cY|$ for any  subset $\cY\subseteq\cF$ of the maximal feasible solutions for~\raf{kp2}, and hence conclude that the latter family can be enumerated in quasi-polynomial time in this case. Note that the same argument does not work if the ordering property~\raf{order} does not hold.
\end{itemize}

\paragraph{Chance-constrained covering binary programs.} In a similar setting as in the previous example, we are given $n$-dimensional normally distributed random vectors $\bw^i\sim N(\ba^i,\Sigma^i)$, for $i\in[r]$, and demands $t_1,\ldots,t_r\in\RR_+$. Here, $w^i_j$ can be interpreted as the coverage value for the $j$th item with respect to the $i$th demand. The requirement is to select a subset of the items, such that the $i$th demand constraint is satisfied with probability at least $\alpha_i\in[0,1]$: 
\begin{alignat}{3}
\label{kp4}
\quad & \displaystyle \Pr[(\bw^i)^\top\bx\geq ~t_i]\ge\alpha_i,\quad\text{ for }i\in[r],\\
&\bx\in\{0,1\}^n.\nonumber
\end{alignat}
As before, we can reformulate the constraints~\raf{kp} as a second-order cone program:
\begin{alignat}{3}
\label{kp5}
\quad & \displaystyle (\ba^i)^\top\bx+\Phi^{-1}(1-\alpha_i)\sqrt{\bx^\top\Sigma^i\bx}\ge t_i,\quad\text{ for }i\in[r]:=\{1,\ldots,r\},\\
&\bx\in\{0,1\}^n.\nonumber
\end{alignat}
In the special case when the random variables are independent, that is, when $\Sigma^i:=(D^i)^2:=\Diag((d^i_1)^2,\ldots,$ $(d^i_n)^2)$ is a full-rank {\it diagonal} matrix, and $\alpha_i\le0.5$ and hence $\Phi^{-1}(1-\alpha_i)\ge0$, we can enumerate the family $\cG$ of minimal feasible solutions for the system~\raf{kp4} by using the result in (B4). Indeed, let $$f_i(\bx):=R_i-(\ba^i)^\top(\b1-\bx)-\Phi^{-1}(1-\alpha_i)\|D^i(\b1-\bx)\|,$$ where $R_i:=(\ba^i)^\top\b1+\Phi^{-1}(1-\alpha_i)\sqrt{\b1^\top\Sigma^i\b1}$, and $\b1$ is $n$-dimensional the vector of all ones.  Then, for any $\bx\in\{0,1\}^n$ s.t. $x_j=0$,
\begin{align}\label{tr2}
f_i(\bx+\b1^j)-f_i(\bx)=a^i_j+\frac{\Phi^{-1}(1-\alpha_i)(d^i_{j})^2}{\sqrt{\sum_{j'\ne j}(d^i_{j'})^2(1-x_{j'})+(d^i_{j})^2}+\sqrt{\sum_{j'\ne j}(d^i_{j'})^2(1-x_{j'})}},
\end{align}
which is monotone increasing $\bx$. Supermodularity follows by Proposition~\ref{p1} below. Next we bound the traction $\tau_i$ and maximum range $R_i$ of each $f_i$. From~\raf{tr2}, we get $\tau_i\ge \min\left\{a_{\min},\Phi^{-1}(1-\alpha_i)\frac {d_{\min}^2}{2\sqrt{n}d_{\max}}\right\}$,
where $a_{\min}=\min_{i,j}\{a^i_{j}~|~a^i_{j}>0\}$, $d_{\min}=\min_{i,j}\{d^i_{j}~|~d^i_{j}>0\}$  and $d_{\max}=\max_{i,j} d^i_{j}$. Similarly, we have $R_i\le na_{\max}+\Phi^{-1}(1-\alpha_i)\sqrt{n}d_{\max}$\add{, where $a_{\max}=\max_{i,j} a^i_{j}$}. It follows from the result in (B4) that if \add{$\frac{\max\{a_{\max},d_{\max}\}}{\min\{a_{\min},d_{\min}\}}$, $\Phi^{-1}(1-\alpha_i)$ and $(\Phi^{-1}(1-\alpha_i))^{-1}$ are bounded by quasi-polynomials} in $n$, then all minimal feasible solutions of the system~\raf{kp4} can be enumerated in quasi-polynomial time.

%\paragraph{Sample selection to minimize maximum variance.} Let $S\in\RR^{d\times n}$ be a matrix whose columns $\bs_1,\ldots,\bs_n\in \RR^d$ represent $n$ samples of data  defined over $d$ features. Given $\bx\in\{0,1\}^n$, the {\it sample mean} and {sample covariance matrix} of the set of samples defined by $\bx$ are given by:
%\begin{align*}
%\hat \bs(\bx)=\frac{\sum_{j=1}^n\bs_jx_j}{\b1^\top\bx}\qquad \sigma(\bx)=\frac{\sum_{j=1}^n\bs_j-x_j}{\b1^\top\bx}
%\end{align*}
%  

\paragraph{Quantum hypergraph covers.}
A quantum hypergraph~\cite{AW02,WX08} is a pair \add{$\cH=(\cV,\cE)$, where $\cV\subseteq\CC^d$ is a $d$-dimensional Hilbert space and each element $j\in \cE:=[n]$ is identified with a Hermitian operator $A_j$ over $\cV$ satisfying $0\preceq A_j\preceq \bI_d$, with $\bI_d$ denoting the $d$-dimensional identity operator over $\cV$ and "$\succeq$`` denoting the {\it L\"owner} (partial) order on Hermitian operators:  $A\succeq B$ if and only if $A-B$ is positive semidefinite. A {\it cover} of $\cH$ is a subset $\cE'\subseteq\cE$ such that $\sum_{j\in\cE'}A_j\succeq \bI_d$.}  This notion arises in the area of quantum information theory~\cite{AW02}. Note that a minimal quantum hypergraph cover is a minimal  feasible solution for the inequality: 
	\begin{alignat}{3}
\label{sdp-cover}
\quad & \displaystyle \sum_{j=1}^nA_jx_j\succeq ~\bI_d,\\
\qquad &\bx\in\{0,1\}^n.\nonumber
\end{alignat}
Assuming feasibility of \raf{sdp-cover}, we must have $T:=\sum_{j=1}^nA_j-\bI_d\succeq0$. It follows then that the minimal quantum hypergraph covers are in one-to-one correspondence with the maximal feasible binary solutions for the inequality $\sum_{j=1}^nA_jx_j\preceq T$, which is of the form considered in (B8). It is not difficult to see that the result in (B8) can  be extended to the case when the matrices $A^{i,j}$ are Hermitian positive semidefinite, while only increasing the  dual bound by a factor of at most $2$ in the exponent (that is, we get $q(\cL,|\cY|)\le O(n)^{4d+1}r|\cY|$). We conclude that, if the dimension $d$ is fixed, then all minimal covers for a quantum hypergraph can be enumerated in quasi-polynomial time.

\section{Supermodular Inequalities}\label{sec:supermodular}

Given a function $f:\cC\to\RR_+$, one can check if $f$ is supermodular using  the following statement, generalizing the well-known characterization of supermodular functions over the Boolean cube \cite{L83}.

\begin{proposition}\label{p1}
	A function $f:\cC\to\RR$ is supermodular if and only if, for any $j\in[n]$, for any $z\in\cC_j\setminus\{c_j\}$, and
	for any $\bx\in\cC_1\times\ldots\times\cC_{j-1}\times\{z\}\times\cC_{j+1}\times\ldots\times\cC_k$,
	the difference
	\begin{align*}
		\partial_f(\bx,j,z) :=f(\bx +\b1^j)-f(\bx),
	\end{align*}
	is monotone in $\bx$.
\end{proposition}
We include the proof in the appendix for completeness.
Define the  \add{``traction"} of $f$, denoted by $\tau(f)$, to be the minimum possible positive increase in $f$ corresponding to a minimal change in the variables \add{along one coordinate}:
\begin{align}\label{traction}
	\tau(f):=\min_{\stackrel{j\in[n],~\bx\in\cC,~x_j<c_j}{f(\bx+\b1^j)>f(\bx)}}f(\bx+\b1^j)-f(\bx).
\end{align}

Consider a monotone inequality~\raf{mon-s}, where the function $f$ is supermodular, and let  $\cF$ denote the family of all its maximal feasible vectors.
We can extend the result in \cite{BEGK-DAM03}  as follows.
\begin{theorem}\label{t1}
	Consider the inequality~\raf{mon-s} and suppose that $f:\cC\mapsto\RR_+$ is a monotone supermodular function with traction $\tau$ and maximum range $R:=f(\bc)$. Then for any subset $\cY\subseteq\cF$ of size $|\cY|\ge 2$, it holds that 
	\begin{align}\label{bd1}
	|\cI(\cY)\cap \cI(\cF)|\leq |\cY|^{o(\log\frac{R-t}{\tau})}.
	\end{align}
	% and $\frac{R-t}{\tau}:=\max\{\max_i \frac{R_i-t_i}{\tau_i},0\}$. %Consequently, if  each $t_i\ge R_i-\poly(\cL)$ and $\frac1\tau_i=\poly(\cL)$, then all maximal feasible vectors for~\raf{mon} can be enumerated in quasi-polynomial time. 
\end{theorem}
\add{Before we prove Theorem~\ref{t1}, we need some preliminaries.} Given $\cY\subseteq\cF$ and $\cX:=\cI(\cY)\cap\cI(\cF)$,
	we follow the proof in \cite{BEGK-DAM03} by constructing a binary tree $\bT$, in which each leaf $l\in L(\bT)$ is mapped to an element $\bx^l\in\cX$, and each internal node $v$ is associated with the element $\bx^v=\bigvee_{l \in L(\bT(v))} \bx^l$; here, $\bT(v)$ denotes the binary sub-tree of $\bT$ rooted at node $v$, and $L(\bT')$ denotes the set of leaves of the subtree $\bT'$. Such a mapping $\phi: L(\bT)\to \cX$ is said to be {\it proper} if it assigns different elements to different leaves, and if $f(\bx^{u} \wedge\bx^{v})\le t$  whenever $u$ and $v$ are incomparable
	nodes of $\bT$ (that is, when the sub-trees $\bT(u)$ and $\bT(v)$ are disjoint). As shown in \cite{BEGK-DAM03}, a sufficiently large binary tree admitting a proper mapping can always be constructed:
	\begin{lemma}[Lemma 11 in \cite{BEGK-DAM03}]\label{l3}
		Let $\cA\subseteq \cC$ be an antichain of size $|\cA|\ge 2$ in an integral box $\cC$ and let $\cB\subseteq\cI(\cA)$. Then there exists a binary tree $\bT$ and a proper mapping $\phi: L(\bT)\to \cB$ such that $|L(\bT)| \geq |\cB|^{1/o(\log |\cA|)}.$ 
	\end{lemma}	
		To prove the theorem, we consider a partition of $\cX=\cX'\cup\cX''$, where $\cX':=\{\bx\in\cX:~f(\bx)\ge t+\frac{\tau}{2}\}\subseteq\cI(\cY)$ and $\cX'':=\cX\setminus\cX'\subseteq\cI(\cY)$, and proper mappings $\phi': L(\bT')\to\cX'$ and $\phi'': L(\bT'')\to\cX''$\add{, defined respectively on two binary trees $\bT'$ and $\bT''$}, as guaranteed by Lemma~\ref{l3}.
	Consequently, the theorem \add{will follow} from the following extension of Lemma~1 in~\cite{BEGK-DAM03}.
	
	\begin{lemma}\label{l1}
		Given binary trees $\bT'$ and $\bT''$ and proper mappings $\phi'$ and $\phi''$ as above, we have 
		\begin{equation}\label{eqp1}
		|L(\bT')|\le \frac{2(R-t)}{\tau}\text{ and }|L(\bT'')|\le \frac{2(R-t)}{\tau}+1.
		\end{equation}
	\end{lemma} 
	\begin{proof}
		Consider first the tree $\bT'$ and the mapping $\phi'$. We show by induction that 
		\begin{equation}\label{p1e1}
		f(\bx^w)\geq t+\frac{\tau}{2}|L(\bT'(w))|.
		\end{equation}
		holds for every node $w$ of the binary tree $\bT'$. Since $f(\bx^w)\le R$, it follows that
		\[
		|L(\bT'(w))|\leq \frac{2(R-t)}{\tau}
		\]
		which, when applied to the root of $\bT'$, proves the first part of the lemma. To see
		\raf{p1e1}, let us apply (backward) induction on the level of the node $w$ in $\bT'$.
		Clearly, if $w=l$ is a leaf of $\bT'$, then $|L(\bT'(l))|=1$, and \raf{p1e1} follows by the assumption that $\bx^l\in\cX'$.
		Let us assume now that $w$ is a node of $\bT'$ with $u$ and $v$ as its immediate successors.
		Then $|L(\bT'(w))|=|L(\bT'(u))|+|L(\bT'(v))|$, and $\bx^w=\bx^u\vee \bx^v$.
		By our inductive hypothesis, and since $f$ is supermodular and $f(\bx^u\wedge \bx^v)\le t$, we have
		the inequalities
		\[
		\begin{array}{rl}
		f(\bx^u\vee \bx^v)&\geq f(\bx^u)+f(\bx^v)-f(\bx^u\wedge \bx^v)\\&\geq
		t+\frac{\tau}{2}|L(\bT(u))|+t+\frac{\tau}{2}|L(\bT(v))|-t\\&=t+\frac{\tau}{2}|L(\bT(w))|.
		\end{array}
		\]
		Consider next the tree $\bT''$ and the mapping $\phi''$. 
		We prove by induction that 
		\begin{equation}\label{p1e2}
		f(\bx^w)\geq t-\frac{\tau}{2}+\frac{\tau}{2}|L(\bT''(w))|.
		\end{equation}
		holds for every node $w$ of the binary tree \add{$\bT''$}. If $w=l$ is a leaf node, then \raf{p1e2} holds as $|L(\bT''(l))|=1$, and  $f(\bx^l)>t$. If $w$ is a node of $\bT''$ with as immediate successors $u$ and $v$, then (as $f(\bx^u)>t$ and $f(\bx^v)>t$ while $f(\bx^u\wedge\bx^v)\le t$), there must exist a $j\in[n]$ such that $\bx^u\wedge\bx^v\le \bx^v-\b1^j$. The definition of $\tau$ and the fact that $\bx^v\in\cX''$ imply that $f(\bx^v-\b1^j)\le f(\bx^v)-\tau<t+\frac\tau2-\tau=t-\frac\tau2$. It follows from this and the inductive hypothesis that
		\[
		\begin{array}{rl}
		f(\bx^u\vee \bx^v)&\geq f(\bx^u)+f(\bx^v)-f(\bx^u\wedge \bx^v)\\&\geq
		f(\bx^u)+f(\bx^v)-f(\bx^v-\b1^j)\\&\geq 
		t-\frac\tau2+\frac{\tau}{2}|L(\bT(u))|+t-\frac\tau2+\frac{\tau}{2}|L(\bT(v))|-(t\add{-}\frac\tau2)\\&=t-\frac\tau2+\frac{\tau}{2}|L(\bT(w))|.
		\end{array}
		\]
		Applying~\raf{p1e2} to the root of $\bT''$ establishes the second part of the lemma.
	\end{proof}
\add{\begin{proof}[Proof of Theorem~\ref{t1}]
	It follows from the above two lemmas that 
	\begin{align*}
	|\cI(\cY)\cap\cI(\cF)|&=|\cX'|+|\cX''|\le |L(\bT')|^{o(\log |\cY|)}+|L(\bT'')|^{o(\log |\cY|)}\\ & \le 2\left(\frac{2(R-t)}{\tau}+1\right)^{o(\log |\cY|)}.
	\end{align*}
\end{proof}}

\add{Note that when $f$ is integer-valued and
$\tau \ge 1$, Theorem~\ref{t1} implies the result in (B1).}
%We give some examples of functions for which Theorems~\ref{t2} and~\raf{t1} can be applied.
\section{Products of Affine Functions}\label{sec:prodcut}
Let $p_1,\ldots,p_m:\cC\to\QQ_+$ be affine functions given in the form: $p_k(\bx)=(\ba^k)^\top\bx+a_0^k$, where $\ba^k\in\QQ_+^n$ are given vectors and $a_k^0\in\QQ_+$ are given numbers.  Given $t\in\QQ_+$, we are interested in enumerating the family $\cG$ of all {\it minimal feasible} vectors for the inequality $g(\bx)\ge t$ over $\bx\in\ZZ_+$, where $g(\bx):=\prod_{k=1}^mp_k(\bx)$. 
We may assume without loss of generality that $\ba^k_j,t\in\ZZ_+$ for all $k,j$.
%$0<t\le \prod_{k=1}^mp_k(\bc)$. 
Although $g(\bx)$ defines a polynomial, we cannot (directly) apply the result in  (B6) since the number of terms $s$ (and possibly also the maximum degree $d$) is exponential in $m$. Instead, we consider the function \begin{align}\label{aux}
	f(\bx):=R-\sum_{k=1}^m\log \overline{p}_k(\bc-\bx),
	\end{align} where $R:=\sum_{k=1}^m\log \overline p_k(\bc)$, $\overline{p}_k(\bx):=p_k(\bx)+\epsilon$, and \begin{align}\label{eps}
	\epsilon:=\frac{1}{2m(1+\max_kp_k(\bc))^{m-1}}
\end{align}
is a sufficiently small perturbation parameter that is needed to ensure that the range of the function $f$ remains bounded (which might fail to hold without perturbation when $\bx=\bc$). 
%, $\tau_k:=\min_{a_j^k>0:~j\in\{0,1\ldots,n\}}a_j^k$ 
%and $p_{-k}(\bc):=\prod_{i\ne k}p_i(\bc)$.
% is bounded by the bit length of the input vectors. 
It is not difficult to see that the family $\cF$ of maximal feasible solutions for the inequality $f(\bx)\le t':=R-\log t$ is in one-to-one correspondence with $\cG$. 
Indeed, given $\bx$ such $g(\bc-\bx)\ge t$, it is immediate that $R-\sum_{k=1}^m\log \overline{p}_k(\bc-\bx)\le t'$, as $\overline p_k(\bc-\bx)\ge p_k(\bc-\bx)$. Conversely, given $\bx$ such that $f(\bx)\le t'$, it holds that $\prod_{k=1}^m\overline p_k(\bc-\bx)\ge t$. As 
\begin{align*}
	\prod_{k=1}^m\overline p_k(\bc-\bx)&=\prod_{k=1}^m p_k(\bc-\bx)+\epsilon\cdot\sum_{S\subseteq[m]~:~|S|<m}\epsilon^{m-1-|S|}\prod_{k\in S}p_k(\bc-\bx)\\
	&\le \prod_{k=1}^m p_k(\bc-\bx)+\epsilon\cdot m(\epsilon+\max_kp_k(\bc))^{m-1},
\end{align*}
we get by our choice~\raf{eps} of $\epsilon$ that $\prod_{k=1}^m p_k(\bc-\bx)\ge t-\frac12$, which in turn implies that $g(\bc-\bx)\ge t$ by the integrality assumption.

To enumerate the elements of $\cF$ in quasi-polynomial time, it would be enough by (B4) to show that the function $f$ is supermodular and to bound both the inverse of the traction $\tau$ of $f$ and the difference $f(\bc)-t'$ by polynomials in the input size. To see that $f$ is supermodular, we apply Proposition~\ref{p1}. For any $\bx\in\cC$ with $x_j<c_j$, we have
\begin{align}\label{e16}
	f(\bx+\b1^j)-f(\bx)=\sum_k\log\frac{\overline{p}_k(\bc-\bx)}{\overline p_k(\bc-\bx-\b1^j)}=\add{\sum_k}\log\left(1+\frac{a_j^k}{p_k(\bc-\bx-\b1^j)+\epsilon}\right),
\end{align}
which is monotone increasing in $\bx$. 

It follows also from \raf{e16} that the traction of $f$ can be bounded from below by
\begin{align*}
	\min_{k,~j~|~a_j^k>0}\log\left(\frac{a_j^k}{(\ba^k)^\top\bc+a_0^k+\epsilon}\right)\add{=\Omega\left(\frac1\cL\right)},
\end{align*}
where $\cL$ is the total encoding length of the coefficients of the given functions. It remains to bound the difference $f(\bc)-t'$, which can be done as follows:
\begin{align*}
f(\bc)-t'&=\log t-\sum_k\log (p_k(\bzero)+\epsilon)\le\log{t}+m\log\left(\frac1\epsilon\right)\\&\le \log t+m^2\log(2m(1+\max_kp_k(\bc)))=O(\cL^2).
\end{align*}
%\begin{align}
%\tau&=\min_{\stackrel{k\in[m],~j\in[n],~\bx\in\cC,~x_j<c_j}{p_i(\bc-\bx)>p_i(\bc-\bx-\b1^j)}}\log\frac{p_i(\bc-\bx)}{p_i(\bc-\bx-\b1^j)}\\&=\min_{\stackrel{k\in[m],~j\in[n],~\bx\in\cC,~x_j<c_j}{a_j^k>0}}\log\left(1+\frac{a_j^k}{(\ba^k)^\top(\bc-\bx-\b1^j)}\right)
%\ge\min_{k\in[m],~j\in[n]}\log\left(1+\frac{a_j^k}{(\ba^k)^\top\bc}\right).\nonumber
%\end{align}

Thus, we arrive at the following result.
\begin{corollary}\label{c3}
	Consider an inequality $\prod_{k=1}^mp_k(\bx)\ge t$ where each $p_k:\cC\to\QQ_+$ is an affine function. Then for any subset $\cX\subseteq \cG$ of the family $\cG$ of minimal feasible solutions of size $|\cG'|\ge 2$, it holds that \add{\footnote{\add{Recall that, for any $\cB\subseteq\cC$, $\cI^{-1}(\cB)$ denotes the family of all maximal non-dominating elements of $\cB$ in $\cC$.}}}
	\begin{align}\label{bd1--}
		|\cI^{-1}(\cX)\cap \cI^{-1}(\cG)|\leq |\cG'|^{o(\log\cL)},
	\end{align}
where $\cL$ is the total encoding length of the coefficients of the given functions.
\end{corollary}	
	
	\add{Note that the introduction of the function~\raf{aux} is merely for the purpose of proving dual-boundedness of the family $\cG$, but is not actually used it the enumeration algorithm.} 
\medskip

It is worth noting that the enumeration  of minimal feasible solutions for a single inequality involving a product of linear functions over binary vectors is as hard as (and hence polynomially equivalent to) the well known {\it hypergraph transversal problem} for which the best currently known algorithm is quasi-polynomial~\cite{FK96}.  Indeed minimal transversals of a given hypergraph $\cH\subseteq 2^{[n]}$ correspond (one-to-one) to the minimal feasible solutions of the inequality
$
\prod_{H\in\cH}\sum_{i\in H}x_i\ge 1, 
$
over $\bx\in\{0,1\}^n$.
\section{Separable Monotone Functions and Polynomial Inequalities}\label{sec:polynomial}

A monotone function $f:\cC\mapsto\RR_+$ is said to be {\it separable} if $f(\bx)$ can be written as the product of single-variable non-negative monotone functions. For instance, a polynomial function of the form \raf{poly} is the sum of separable monotone functions. A single-variable function $g:\add{\{0,1,\ldots,c\}}\to\RR_+$ is said to be {\it discretely convex} if for all $x,y\in\add{\{0,1,\ldots,c\}}$ and $\lambda\in[0,1]$,
\begin{align*}
	\lambda g(x)+(1-\lambda) g(y) \ge \min_{u\in\add{\{0,1,\ldots,c\}}:~|\lambda x+(1-\lambda)y-u|\le 1}g(u).
\end{align*}
A sufficient and necessary condition for discrete  convexity is that (see, e.g., \cite{Y02})
the difference $\partial_g(x)=g(x+1)-g(x)$ is monotone in $x\add{\in\{0,1,\ldots,c-1\}}$. 

Let $f:\cC\to\RR_+$ be the sum of separable monotone functions, that is, 
\begin{align}\label{sep}
f(\bx)=\sum_{H\in\cH}a_H\prod_{j\in H}f_j^H(x_j),\quad\text{ for $\bx\in\cC$},
\end{align}
 where $\cH\subseteq 2^{[n]}$ is a given multiset family (repetitions allowed),  $f_j^H:\ZZ_+\mapsto\RR_+$ are non-negative monotone functions and $a_H>0$ for all $H\in\cH$. 
 
\medskip
	 
Applying Proposition~\ref{p1} to the function $f$ defined in \raf{sep}, we see that a sufficient condition for $f$ to be supermodular is that the difference function \add{$\partial_{f_j^H}(x_j)$} is monotone in $x_j$, or equivalently, $f_j^H$ is discretely convex, for all $j\in[n]$, as
\begin{align*}
	\partial_f(\bx,j,z)=\sum_{H\in\cH~|~j\in H}a_H\big(f_j^H(x_j+1)-f_j^H(x_j)\big)\prod_{j'\in H\setminus\{j\}}f_{j'}^H(x_{j'})
\end{align*}
is monotone whenever \add{$\partial_{f_j^H}(x_j)=f_j(x_j+1)-f_j(x_j)$} is monotone in $x_j$. As a corollary of this and the bound mentioned in (B1) above, we obtain the following result.
\begin{proposition}\label{p2}
Consider the inequality~\raf{mon-s} and suppose that the function $f:\cC\mapsto\RR_+$ has traction $\tau$ and maximum value $R:=f(\bc)\ge t$, and is the sum of discretely convex separable monotone functions of the form~\raf{sep}, defined by a multiset family $\cH$. Then for any $\cY\subseteq\cF$ of size $|\cY|\ge 2$, it holds that 
\begin{align}\label{bd1-}
|\cI(\cY)\cap \cI(\cF)|\leq |\cY|^{o(\log \frac{R-t}{\tau})}.
\end{align}
% Consequently, if  each $t_i\ge R_i-\poly(n,r,s)$  and $\tau_i=\poly(n,r)$, then all maximal feasible vectors for~\raf{mon} can be enumerated in quasi-polynomial time.
\end{proposition}

When the function $f$ is integer-valued,
we can drop the requirement of discrete convexity. In fact, in this case, the bound in Proposition~\ref{p2} can be improved as follows.
\begin{theorem}\label{t1-}
	Consider the inequality~\raf{mon-s} and suppose that  $f:\cC\mapsto\{0,1,\ldots,R\}$ is the sum of separable monotone functions of the form~\raf{sep}, defined by a multiset family $\cH$, such that each $f_j^H:\cC_j\mapsto\{0,1,\ldots,R\}$ is monotone for all $j\in H\in\cH$ and $a_H\in\ZZ_+$ for all $H\in\cH$. Then for any non-empty subset $\cY\subseteq\cF$, it holds that 
	\begin{align}\label{bd2}
	|\cI(\cY)\cap \cI(\cF)|\leq (R-t)|\cY|.
	\end{align}

\end{theorem}
To prove Theorem~\ref{t1-}, we consider the class of functions having non-negative 
{\em M\"obius} coefficients. Recall that the 
M\"obius function $\mu:\cC\times\cC\mapsto\{-1,0,1\}$ is given by (see, e.g., \cite {BBMW88}):
\begin{equation}\label{ch3:mob}
\mu(\by,\bz)=\left\{
\begin{array}{ll}
(-1)^{|S|} & \mbox{ if \add{$\by=\bz-\b1^S$} for some $S\subseteq[n]$ }\\
0 & \mbox{ otherwise}
\end{array}
\right. 
\end{equation}
for $\by,\bz\in\cC$, where $\b1^S\in\{0,1\}^n$ is the vector having 
$\b1_j^S=1$ if and only if $j\in S$. 
Given a function $f:\cC\mapsto\RR_+$ and an $\bx\in\cC$, the {\em M\"obius inversion formula} 
enables us to express $f(\bx)$ as the sum of {\em M\"obius coefficients} $\hat{f}(\by)$ of all
elements $\by\leq\bx$:
\begin{equation}\label{ch3:mob-inv}
f(\bx)=\sum_{\bzero\leq \by\leq\bx}\hat{f}(\by)~\Longleftrightarrow~\hat{f}(\by)=\sum_{\bzero\le \bz\leq \by}f(\bz)\mu(\bz,\by).
\end{equation}

\begin{proposition}\label{p4}
	Suppose $f:\cC\mapsto\{0,1,\ldots,R\}$ is the sum of separable monotone functions: $f(\bx)=\sum_{H\in\cH}a_H\prod_{j\in H}f_j^H(x_j)$, where $f_j^H:\cC_j\mapsto\{0,1,\ldots,R\}$ is a monotone function and $a_H\ge0$ for all $H\in\cH$. Then $\hat f(\bx)\ge 0$ for all $\bx\in\cC$. %Furthermore, if $\hat f_i(\bx)>0$, then $\hat f_i(\bx)\ge1$.
\end{proposition}
\begin{proof}
	Consider a term $g_H(\bx):=\prod_{j\in H}f_j^H(x_j)$. \add{If $x_j>0$ for some $j\not\in H$, we have $\hat g_H(\bx)=0$. Otherwise,}
$$
\hat{g}_H(\bx)=\sum_{S\subseteq H}(-1)^{|S|}g_H(\bx-\b1^S)=\sum_{S\subseteq H}(-1)^{|S|}
\prod_{j\in H} f_j^H(\bx-\b1^S)=\prod_{j\in H}[f_j^H(x_j)-f_j^H(x_j-1)],
$$ 
for any $\bx\in\cC$, where we assume $f_j^H(-1):=0$ for all $j\in H$. The non-negativity of $\hat{g}_H(\bx)$ follows form the monotonicity of $f_j^H$. By the linearity of the M\"obius transform, it follows then that $\hat f(\bx)=\sum_{H\in\cH_i}a_H\hat g_H(\bx)\ge 0$. 
%If $\hat f_i(\bx)>0$, then there exits $H\in\cH_i$ such that $\hat g_H>0$. By the integrality of $f_j^H$, $f_j^H(x_j)-f_j^H(x_j-1)\ge 1$ for all $j\in H$, implying that $\hat f_i(\bx)\ge\hat g_H(\bx)\ge1$.
\end{proof}

\add{Theorem~\ref{t1-} }follows from the following intersection inequality from
\cite{BGKM04}:
\begin{lemma}[Intersection Lemma \cite{BGKM04}]
	\label{int-lem}
	Let $\cS,\cT\subseteq 2^U$ be two families of subsets of a given set
	$U$, and $w:U \rightarrow\mathbb{R}_+$ be a given non-negative weight function
	on $U$. Suppose $\cS$ and $\cT$ are threshold separable, i.e., there are real thresholds $t_1<t_2$, such that $w(T)\leq t_1$, for all
	$T\in\cT$, and $w(S)\geq t_2$, for all $S\in\cS$, where $w(X)=\sum_{v\in X}w(v)$ for
	$X\subseteq U$. Suppose further that $|\cS|\geq 2$ and $\cT$ covers
	all pairwise intersections of $\cS$, i.e., for
	all $S,S'\in\cS$, $S\neq S'$, there exists a $T\in\cT$ such that
	$S\cap S'\subseteq T$. Then 
		\smallskip
	
	$(i)$~~   $|\cS| \leq \sum_{T \in \cT}|U\setminus T|$,
	
	\smallskip
	
	$(ii)$~  
$|\cS| \leq \frac{w(U)-t_1}{t_2-t_1}|\cT|$.
	
\end{lemma} 

The proof of part (i) of Lemma~\ref{int-lem} was given in \cite{BGKM04}. The proof of part (ii) for the unweighted case (i.e., $w(u)=1$ for all $u\in U$) was given (for a weaker inequality) in \cite{BGKM00}. We include the proof of the weighted case of part (ii) in the appendix for completeness.

\bigskip

\begin{proof}[\add{Proof of Theorem~\ref{t1-}}]
Now to prove the theorem, we let $\cX\defeq\cI(\cY) \cap  \cI(\cF)$ and
consider the (one-to-one) monotonic mapping $\phi:\cC\mapsto 2^\cC$ 
defined by: $\phi(\bx)=\{\bz\in\cC~|~\bz\leq \bx\}$. Let $U:=\cC$, 
$\cS:=\{\phi(\bx)~|~\bx\in\cX\}$, and  $\cT:=\{\phi(\by)~|~\by\in\cY\}$. 
Thus with respect to the  non-negative weight function
$w\equiv\hat{f}:U\to\RR_+$, we obtain the threshold separability 
\begin{align}\label{e10}
\begin{array}{ll}
w(\phi(\bx))&=\sum_{\bzero\le\bz\le\bx}\hat f(\bz)=f(\bx)\geq t_2:=t+1,\mbox {for all  }\bx\in\cX;\\
w(\phi(\by))&=\sum_{\bzero\le\bz\le\by}\hat f(\bz)=f(\by)\leq t_1:=t,\mbox {for all  }\by\in\cY, 
\end{array}
\end{align}  
of $\cS$ and $\cT$. %Note that, by Proposition~\ref{p4},  for all $\bz\in U$,  either if $w(\bz)=\hat f(\bz)=0$, or $w(\bz)\ge1$. 
If $|\cX|=|\cS|=1$, then \add{as $f(\bc)\geq f(\bx)\geq  t+1$ for the element $\bx\in\cX$, we get $|\cX|=1\le f(\bc)-t\le (f(\bc)-t)|\cY|$}.
Let us assume therefore that $|\cS|\geq 2$, and observe that $\cT$ covers all pairwise intersections of $\cS$: for any two distinct elements $\bx,\bx'\in\cX$,
it follows by $\bx,\bx'\in\cI(\cY)$ that there is a $\by\in\cY$ such
that $\bx\wedge \bx' \leq \by$, and therefore, we get 
$$
\phi(\bx)\cap\phi(\bx')=\phi(\bx\wedge \bx')\subseteq\phi(\by).
$$
Now we apply Lemma~\ref{int-lem}(ii) to get   
$$
\left|\cI(\cY) \cap  \cI(\cF)\right|=|\cS| \leq \left(\sum_{\bzero\le\bz \le \add{\bc}}\hat{f}(\bz)-\add{t}\right)|\cT|=(f(\add{\bc})-t)|\cY|.
$$
\end{proof}
\add{Note that the restriction that $f$ is integer-valued was required in the proof of Theorem~\ref{t1-} only to gurantee that $t_2-t_1\ge 1$. One can easily see that we may drop this restriction and replace the right-hand side of the bound in \raf{bd2} by $\frac{R-t}{\kappa(f)}|\cY|$, where $\kappa(f):=\min_{\bx,\by\in\cC:~f(\bx)<f(\by)}(f(\by)-f(\bx))\le\tau(f)$. However, it is not clear if a bound of the form~\raf{bd2} with the right-hand side replaced by $\frac{R-t}{\tau(f)}|\cY|$ can be obtained.}

\medskip

Theorem~\ref{t1-} implies that, if the threshold $t$ is sufficiently close to the maximum value of $f$, then the enumeration problem can be solved in quasi-polynomial time.
In the other extreme case, when each $t=\poly(\cL)$ is sufficiently far form the maximum value, we can improve the bound in Theorem~\ref{t1-} (with slightly some more restrictions on the functions  $f_j^H$) as follows.
\begin{theorem}\label{t3}
		Consider the inequality~\raf{mon-s} and suppose that  $f:\cC\mapsto\{0,1,\ldots,R\}$ is the sum of separable monotone functions  of the form~\raf{sep}, defined by a (non-empty) multiset family $\cH$, such that each $f_j^H:\cC_j\mapsto\{0,1,\ldots,R\}$ is monotone with $f_j^H(0)=0$ and $f^H_j(1)\ge 1$, for all $j\in H\in\cH$ and $a_H\in\ZZ_+$ for all $H\in\cH$. Then for any non-empty subset $\cY\subseteq\cF$, it holds that 
	\begin{align}\label{bd2-}
		|\cI(\cY)\cap \cI(\cF)|\leq \left(|\cH|+\add{\big(1+\rho n\big)}\add{t|\cY|}\right) |\cY|,
	\end{align}
	where $\rho:=\max_{j\in H\in\cH,~x\in\cC_j\setminus\add{\{0,c_j\}}}\frac{\add{f_j^H(x+1)}}{f_j^H~(x)}$.
\end{theorem}
\begin{proof}
    Let \add{$\cX\defeq\cI(\cY) \cap  \cI(\cF)$}.
    %, where $\cF$ is the set of maximal feasible solutions of the inequality $f(\bx)\le t$
     We consider a partition of $\cX:=\cX_1\cup\cX_2$, where $\cX_1:=\cX\cap\{0,1\}^n$ and $\cX_2:=\cX\setminus\cX_1$, and define $\cY_1:=\{\by\wedge\bone~|~\by\in\cY\}$, where $\bone$ is the vector of all ones. We first bound the size of $\cX_1$. Let $U:=\cH$, $\cS:=\{\phi(\bx)~|~\bx\in\cX_1\}$, and  $\cT:=\{\phi(\by)~|~\by\in\cY_1\}$, where the monotonic mapping $\phi:\cC\to 2^{\cH}$ is defined by: $\phi(\bx):=\{H\in\cH~|~\bone^H\le\bx\}$\add{, with $\bone^H$ denoting the characteristic vector of $H\subseteq[n]$}. Note that our assumptions imply that,  for $\bx\in\{0,1\}^n$, 
	\begin{align*}
	f(\bx)&=\sum_{H\in\cH}a_H \prod_{j\in H}f_j^H(1)\prod_{j\in H}x_j
	=\sum_{H\in\cH~:~\bone^H\le \bx}a_H \prod_{j\in H}f_j^H(1).
	\end{align*} 
	Thus with respect to the  non-negative weight function
	$w:U\to\RR_+$, defined by $w(H):=\add{a_H\prod_{j\in H}f_j^H(1)}$ for $H\in\cH$, we obtain the threshold separability 
	\begin{align*} 
	w(\phi(\bx))&=\add{f(\bx)}\ge t+1,\mbox { for all  }x\in\cX_1;~~~~~w(\phi(\by))=\add{f(\by)}\leq t,\mbox {for all  }y\in\cY_1,
	\end{align*}  
	of $\cS$ and $\cT$. 
	Observe that $\cT$ covers all pairwise intersections of $\cS$: for any two distinct elements $\bx,\bx'\in\cX_1$,
	it follows by $\bx,\bx'\in\cI(\cY)\cap\{0,1\}^n$ that there is a $\by\in\cY_1$ such
	that $\bx\wedge \bx' \leq \by$, and therefore, we get 
	\begin{align}\label{e4}
	\phi(\bx)\cap\phi(\bx')=\phi(\bx\wedge \bx')\subseteq\phi(\by).
	\end{align}
	Threshold separability together with~\raf{e4} also implies that $|\cS|=|\cX_1|$.
	If $|\cS|=|\cX_1|=1$, then $|\cS|\le |\cH|\cdot|\cY|$ holds by the non-emptiness of $\cH$ and $\cY$. Let us assume therefore that $|\cS|\geq 2$, and 
	apply Lemma~\ref{int-lem}(i) to get    
	\begin{align}\label{e8}
	\left|\cX_1\right|=|\cS| \leq \sum_{\by \in \cY_1}|\cH\setminus \phi(\by)|\add{\le|\cH|\cdot|\cY_1|\le|\cH|\cdot|\cY|}.
	\end{align}
	
	Next, we will show that
	\begin{align}\label{e9}
	\left|\cX_2\right| \le \big(1+\add{\rho n\big)t|\cY|^2},
	\end{align}
	which together with \raf{e8} would imply the theorem. To see~\raf{e9}, we 
	consider the \add{(one-to-one) }monotonic mapping $\phi:\cC\mapsto 2^\cC$ 
	defined by: $\phi(\bx)=\{\bz\in\cC~|~\bz\leq \bx\}$. Let
	$\cS:=\{\phi(\bx)~|~\bx\in\cX_2\}$ and $\cT:=\{\phi(\by)~|~\by\in\cY\}$. By definition of $\cX_2$, for any $\bx\in\cX_2$, there exists a $j:=j^\bx\in[n]$ such that $x_j>1$. As $\cX_2\subseteq\add{\cI(\cY)}$, for any $\bx\in\cX_2$, there is a \add{$\by:=\by^\bx\in\cY$} satisfying $\bx-\b1^{j^\bx}\le \by$. It is important for the following argument to note that $y^\bx_{j}\ge 1$ (and hence $f^H_j(y_j^\bx)\ge 1$), for any $\bx\in\cX_2$ and $j:=j^\bx$, as this	implies that, for $\by:=\by^\bx$,
	\begin{align}\label{ee-}
	f(\by+\b1^{j})&= f(\by)+\sum_{H\in\cH~|~j\in H}a_H\big(f_j^H(y_j+1)-f_j^H(y_j)\big)\prod_{j'\in H\setminus\{j\}}f_{j'}^H(y_{j'})\nonumber\\
	&= f(\by)+\sum_{H\in\cH~|~j\in H}a_H\frac{f_j^H(y_j+1)-f_j^H(y_j)}{f_j^H(y_j)}\prod_{j'\in H}f_{j'}^H(y_{j'})\nonumber\\
	&\le f(\by)+\add{(\rho-1)\sum_{H\in\cH~|~j\in H}a_H \prod_{j'\in H}f_{j'}^H(y_{j'})\le f(\by)+(\rho-1) f(\by)\le\rho t}.
	\end{align}
	Define $\cY_2:=\{\by^\bx+\bone^{j^\bx}~|~\bx\in\cX_2\}$, and  $U:=\{\bz\in\cC~|~\bz\le\by\text{ for 
	 }\by\in\cY\cup\cY_2\}$. The definition of $U$ and the monotonicity of the mapping $\phi$ imply that $\phi(\bx),\phi(\by)\subseteq U$ for all $\bx\in\cX_2$ and $\by\in\cY$. Note also that $|\cY_2|\le n|\cY|$.
	
	Now we proceed in a way similar to the proof of Theorem~\ref{t1-}. We apply Lemma~\ref{int-lem}(ii) using the non-negative weight function
	$w\equiv\hat{f}:U\to\RR_+$, and the threshold separability~\raf{e10} of $\cS$ and $\cT$ \add{(where $\cX$ is replaced by $\cX_2$)}.
	If $|\cX_2|=|\cS|=1$, then \raf{e9} holds trivially. Thus we may assume that $|\cS|\ge 2$, and observe again that $\cT$ covers all pairwise intersections of $\cS$. Applying the lemma and using~\raf{ee-}, we obtain
	\begin{align*}
	\left|\cX_2\right|&=|\cS| \leq \left(\sum_{\bz \in U}\hat{f}(\bz)-\add{t}\right)|\cT|\le\left(\sum_{\by\in\cY}\sum_{\bzero\le\bz \le \by}\hat{f}(\bz)+\sum_{\by\in\cY_2}\sum_{\bzero\le\bz \le \by}\hat{f}(\bz)-t\right)|\cT|\\
	&=\left(\sum_{\by\in\cY}{f}(\by)+\sum_{\by\in\cY_2}f(\by)-t\right)|\cT|\le
	\left(t|\cY|+\add{\rho t}|\cY_2|-t\right)|\cT|\\
	&\le \big((1+\add{\rho n)}|\cY|-1\big)t|\cY|,
	\end{align*}
	establishing \raf{e9}.
\end{proof}

Applying Theorems~\ref{t1-} and~\ref{t3} to a polynomial function of the form~\raf{poly} with $f_j^H(x_j):=x_j^{d_{H,j}}$ and noting that $\max_{x\in\cC_j\add{\setminus\{0,c_j \}}}\frac{\add{f_j^H(x+1)}}{f_j^H~(x)}\add{=}\max_{x\ge 1}\frac{(x+1)^{d_{H,j}}}{x^{d_{H,j}}}=2^{d_{H,j}}$, we arrive at the following result.

\begin{corollary}\label{c2} 
			Consider the inequality~\raf{mon-s} and suppose that the function $f:\cC\mapsto\{0,1,\ldots,R\}$ is a polynomial of the form~\raf{poly}, defined by a (non-empty) multiset family $\cH$, and coefficients $a_H\in\ZZ_+$ for $H\in\cH$. Then for any non-empty subset $\cY\subseteq\cF$, it holds that 
	\begin{align}\label{bd2--}
	|\cI(\cY)\cap \cI(\cF)|\leq \min\left\{R-t,|\cH|+\big(1+2^dn)\add{t|\cY|}\right\} |\cY|.
	\end{align}
	where $d:=\max_{j\in H\in\cH}d_{H,j}$.% In particular, if $d=O(1)$ and either $t_i\ge R_i-\quasipoly(n,r,s)$  or $t_i=\quasipoly(n,r,s)$, then all maximal feasible vectors for~\raf{mon} can be enumerated in quasi-polynomial time.
\end{corollary}

Note that, unlike the result in Corollary~\ref{c2}, the result in (B3) requires the {\it total} degree in each monomial to be bounded by a constant (but without any restriction  on $t$) to guarantee quasi-polynomial enumeration.  

\add{Corollary~\ref{c2} can be complemented } with the following negative result \add{from~\cite{BEGK02-MFCS}}, which shows that the polynomial dependence on \add{$R$ and }$|\cH|$ in the bound \raf{bd2--} is necessary. \add{The proof is included in the appendix for completeness.} 
\begin{proposition}[\add{Based on~\cite{BEGK02-MFCS}}]\label{ch1:unbdd}
	There exists a polynomially computable polynomial function $f:\{0,1\}\to\add{\ZZ_+}$\add{ of the form~\raf{poly}, with exponentially large $R=f(\bone)$ and $|\cH|$},
	for which problem  {\em GEN}$(\cY)$ of incrementally generating the maximal feasible solutions of the inequality $f(x)\le 0$ is NP-hard.
\end{proposition}

%Note that the size of $\cH$ in the NP-hardness construction above is exponential  in \add{the size of the input}.  
Whether  a bound of the form~\raf{bd2--} that is independent of $t$ \add{ and $R$} (as in the Boolean case \cite{BGKM04}, where it is known that $|\cI(\cY)\cap \cI(\cF)|\leq |\cH|\cdot|\cY|$) exists,
remains an interesting open question.

\section{Second-order Cone Inequalities}\label{sec:soc}
For a vector $\bv\in\RR^n$, denote by $\|\bv\|$ the $\ell_2$ norm of $\bv$. Consider the monotone second-order cone \add{(SOC)} inequality:
\begin{alignat}{3}
	\label{socp}
	\quad & \displaystyle f(\bx):= \|A\bx\|+\bb^\top\bx\leq ~t,
\end{alignat} 
where $A\in\RR_+^{d\times n}$ and $\bb\in\RR_+^n$ are  given matrix and vector. In the special case when $\bb=\bzero$, we can derive the following bound using (B3) as \raf{socp} reduced to a quadratic inequality.

\begin{proposition}\label{p3}
	Let $\cF$ be the set of maximal feasible vectors for~\raf{socp}, where $\bb=\bzero$. Then, for any non-empty $\cY\subseteq\cF$, it holds that 
	\begin{align}\label{bd3}
		|\cI(\cY)\cap \cI(\cF)|\leq \add{n(n+1)}(2|\cY|+1)^2.
	\end{align}
\end{proposition}
\begin{proof}
	When $\bb=\bzero$, we can write the inequality in~\raf{socp} as $(f(\bx))^2\le t^2$, where $(f(\bx))^2=\bx^\top A^\top A\bx$ is a quadratic function with non-negative coefficients. It follows from (B3) (with \add{$s=n(1+1)/2$} and $p=2$) that $|\cI(\cY)\cap \cI(\cF)|\leq \add{n(n+1)}(2|\cY|+1)^2$.
\end{proof}
When $\bb\neq \bzero$, the above argument does not work\footnote{Indeed, squaring does not yield an equivalent problem as taking the square root results in two possibilities (e.g.,  consider $\sqrt{x_1+x_2}+2x_1\le 1$; squaring yields $x_1+x_2\le 1$ which is not an equivalent inequality). Moreover, squaring both sides of an inequality like~\raf{socp} may yield a term with a negative coefficient (indeed we get $\bx^\top A^\top A\bx+2t\bb^\top\bx-(\bb^\top\bx)^2\le t^2$), where the result in (B3) cannot be applied (e.g., consider $\sqrt{x_1+x_2}+x_1+x_2\le 2$; squaring yields $2x_2+2x_2-x_1x_2\le 2$ which is an equivalent inequality but with a negative coefficient).}. To bound the number of \add{minimal} infeasible vectors for \raf{socp}, we use a different argument based on a {\it semi-infinite linear} formulation of \raf{socp}. 

Denote by $\add{\BB^d_+}:=\{\bx\in\RR_+^d:~\|\bx\|\le 1\}$ the non-negative half of the $d$-dimensional unit ball centered at the origin. We can rewrite \raf{socp} in the following equivalent form:
\begin{alignat}{3}
\label{socp2}
\quad & \displaystyle f_{\bu}(\bx):= \bu^\top A\bx+\bb^\top\bx\leq ~t,\quad\text{ for }\bu\in\BB^{d}_+(0,1).
\end{alignat} 

As~\raf{socp2} is a (semi-infinite) monotone system of linear inequalities, we may be tempted to apply the result in (B2) for a single inequality and then take a union bound. However, as the number of inequalities in~\raf{socp2} is infinite and the  union is taken over an {\it uncountable} set, the union bound does not hold\footnote{In fact, a simple but incorrect proof via the union bound can go as follows. Using the notation in the proof of Theorem~\ref{t4}, we have \add{$|\cI(\cY) \cap  \cI(\cF)|=|\bigcup_{\bu\in\BB^{d}_+(0,1)}\cI(\cY) \cap  \cI(\cF_{\bu})|\le \int_{\BB^{d}_+(0,1)}|\cI(\cY) \cap  \cI(\cF_{\bu})|d\bu\le n|\cY|\vol(\BB^{d}_+(0,1))\sim\frac{n|\cY|}{\sqrt{\pi d}}\left(\frac{2\pi e}{d}\right)^{d/2}$}. This (incorrect) bound is counter-intuitive in the sense that it decreases with $d$ (for large $d$).}. Instead, we argue that we can take the  union bound only over \add{$O(n)^{2d}$} inequalities. To see this, we first recall the following definition and a lemma. 

A monotone function  $f:\cC \rightarrow  \RR_+$
is called {\em 2-monotonic} if there exists a permutation $\sigma:[n]\to[n]$ such that, for all $\bx\in\cC$ and $k,j\in[n]$ with $k<j$, $\bx_{\sigma(k)}<c_{\sigma(k)}$ and $\bx_{\sigma(j)}>0$, it holds that $f(\bx+\b1^{\sigma(k)}-\b1^{\sigma(j)}) \ge f(\bx)$. For instance, if $f(\bx):=\sum_{j}^nw_jx_j$ is a linear function with non-negative coefficients (i.e., $\bw\ge 0$), then $f$ is 2-monotonic (with $\sigma=\sigma_\bw$ being a permutation satisfying $w_{\sigma(1)}\ge w_{\sigma(2)}\ge\cdots\ge w_{\sigma(n)}$). 

\begin{lemma}[Based on~\cite{C87,BEGKM02-SICOMP}]\label{l2}
    Consider the system~\raf{mon} \add{but where the inequality index $i$ may vary over an {\it uncountable} set $U$} and suppose that each function $f_i:\cC\mapsto\RR_+$ is a 2-monotonic function as verified by a permutation $\sigma_i:[n]\to[n]$. Let $\cF$ be the set of maximal feasible vectors for~\raf{mon}. Then for any non-empty subset $\cY\subseteq\cF$, it holds that 
	\begin{align}\label{2mon-dbdd}
		|\cI(\cY)\cap \cI(\cF)|\leq r'\sum_{\by\in\cY}q(\by),
	\end{align}
	where \add{$r':= 
	|\{\sigma_i|~i\in U\}|$ is the number of distinct permutations among the $\sigma_i$'s} and $q(\by):=|\{j\in[n]~:~y_j<c_j\}|$.  
\end{lemma}
We give the proof in the appendix for completeness. 
We will also need the following geometric fact.
\begin{fact}[see, e.g., \cite{M02}]\label{f1}
Any arrangement of $m$ $d$-dimensional hyperplanes partitions $\RR^d$ into at most $\Phi_d(m):=\sum_{i=0}^d \binom{m}{i}\le \left(\frac{e m}{d}\right)^d$ maximal  connected regions not intersected by any  of the hyperplanes (called  cells of the arrangement).  
\end{fact}	
%Note that if all hyperplanes in an arrangement  pass through a common point, then it is easy to derive the (stronger) bound $2\psi(m-1,d-1)$ on the number of cells of the arrangement, using Fact~\ref{f1}.

\medskip

We are now ready to prove \add{the following theorem}.

\begin{theorem}\label{t4}
Let $\cF$ be the set of maximal feasible vectors for~\raf{socp}. Then for any non-empty subset $\cY\subseteq\cF$, it holds that 
\begin{align}\label{socp-bdd}
|\cI(\cY)\cap \cI(\cF)|\leq \Phi_{d}(n(n-1)/2)n|\cY|=O(n)^{2d+1}|\cY|.
\end{align}
%Consequently, if  $d=O(1)$, then all maximal feasible vectors for~\raf{socp} can be enumerated in quasi-polynomial time.
\end{theorem}
\begin{proof}
For $\bu\in\add{\BB_+^{d}}$, let $\bw^{\bu}:=A^\top\bu+\bb\in\RR_+^n$ and $\cF_{\bu}$ be the set of maximal feasible solutions for the inequality $f_{\bu}(\bx)=(\bw^{\bu})^\top\bx\le \add{t}$. \add{Then, $\cI(\cY) \cap  \cI(\cF)=\bigcup_{\bu\in \BB_+^d}^r\left(\cI(\cY) \cap  \cI(\cF_\bu)\right)$.} By Lemma~\ref{l2}, it is enough to show that the number of distinct permutations defined by the set of weights  
$\{\bw^{\bu}~|~\bu\in\add{\BB^{d}_+}\}$ is at most $\Phi_{d}(n(n-1)/2)$. More precisely, to each vector $\bw\in\RR^n_+$, let us assign a permutation  $\sigma=\sigma_\bw$ satisfying $w_{\sigma(1)}\ge w_{\sigma(2)}\ge\cdots\ge w_{\sigma(n)}$ (note that that there may be multiple permutations $\sigma$ satisfying this, in which case $\sigma_\bw$ is chosen arbitrarily among them). Then we claim that
\begin{align}\label{e12}
|\{\sigma_{\bw^{\bu}}~|~\bu\in\add{\BB^{d}_+}\}|\le \Phi_{d}(n(n-1)/2).
\end{align}
To see~\raf{e12}, let us write $A=[\ba^{1},\ldots,\ba^{n}]$ where $\ba^{j}\in\RR^d_+$ is the $j$th column of $A$. Then $w_j:=w_j^{\bu}=(\ba^{j})^\top\bu+b_j$. Let us consider the system of inequalities $w_j\le w_{j'}$ for\add{ distinct} $j,j'\in[n]$ (considering $\bu$ as a variable in $\RR^{d}$):
\begin{align}\label{e13}
\left(\ba^{j}-\ba^{j'}\right)^\top\bu\le b_{j'}-b_j,\quad\text{ for }j\ne j'\in[n].
\end{align}
The inequality-defining hyperplanes in~\raf{e13} form a hyperplane arrangement that, by Fact~\ref{f1}, partitions $\RR^{d}$ into at most $\Phi_{d}(n(n-1)/2)$ cells. Consider any such cell $C$. Any point $\bu\in C$ decides, for each pair $j\ne j'$, whether $w^{\bu}_j\le w^{\bu}_{j'}$ or $w^{\bu}_j>w^{\bu}_{j'}$, and hence can be associated with a total order on the weights $w^{\bu}_1,\ldots, w^{\bu}_n$. Moreover, all points in $C$ give rise to the same total order, while any two points belonging to two different cells give rise to two different orders. It follows that  the number of such orders is exactly equal to the number of cells. This  establishes~\raf{e12} and the theorem.
\end{proof}
\add{
\begin{remark}
The proof of Theorem~\ref{t4} shows that, for the purpose of enumerating the set $\cI(\cY)\cap\cI(\cF)$ for a given subset $\cY\subseteq\cF$ of the maximal feasible solutions of a SOC inequality~\raf{socp}, we may replace the inequality~\raf{socp} by $O(n)^{2d}$ linear inequalities obtained by selecting a representative $\bu\in\BB_+^d$ from each cell of the arrangement determined by~\raf{e13}. One should note, however, that this linear system is not equivalent to~\raf{socp} in terms of enumerating the family $\cF$.\footnote{\add{For example, consider the two inequalities $x_1+x_2\le 1$ and $x_1+2x_2\le 1$. While they share the same permutation, their sets of maximal feasible vectors are different.}}  
\end{remark}
}
\section{Positive Semidefinite Matrix Inequalities}\label{sec:psd}
We denote by $\SS^m$ the set of all $m\times m$ real symmetric matrices and by  
$\SS^m_+\subseteq \SS^m$ the set of all $m\times m$ positive semidefinite matrices. 
Consider the monotone positive semidefinite \add{(PSD)} inequality:
		\begin{alignat}{3}
		\label{sdp}
	    \quad & \displaystyle f(\bx)\preceq ~T,\\
		\qquad &\bx\in\ZZ_+^n,\nonumber
		\end{alignat}
\noindent where $f(\bx):=\sum_{j=1}^nA^{j}x_j$, $A^{j}\in\SS_+^{m}$, for  $j\in[n]$, and $T\in\SS^{m}_+$ are given positive semidefinite matrices, and
"$\succeq$`` is the {\it L\"owner} (partial) order on matrices:  $A\succeq B$ if and only if $A-B$ is positive semidefinite. %This type of SDPs arise in many applications, see, e.g. \cite{IPS11,IPS10} and the references therein.

Let $\bI_m$ be the $m\times m$ identity matrix. For two matrices \add{$A, B\succeq 0$}, we use the standard notation:  $A\bullet B=\Tr(AB):=\sum_{k=1}^m\sum_{j=1}^m a_{k,j}\add{b_{k,j}}$, where $a_{kj}$ denotes the $kj$-th entry of the matrix $A$. 
%We denote by $\lambda_{\min}(A):=\lambda_1(A)\le\lambda_2(A)\le\cdots\le\lambda_m(A)=:\lambda_{\max}(A)$ the eigenvalues of $A$ in increasing order. 
We recall the following well-known facts; see,.e.g., \cite{HJ90}:

\begin{fact}\label{f2}
	Let $A\in\SS^m$. Then 
	\begin{itemize}
		\item[(i)] $A\in\SS_+^m$ iff $A\bullet B\ge 0$ for all $B\in\SS_+^m$;
		\item[(ii)] if $A\in\SS_+^m$ then for any $i\in[m]$, $a_{ii}\ge 0$ with $a_{ii}=0$ implying that the entire $i$th row and \add{column} of $A$ are zero;
		\item[(iii)] $A\in\SS_+^m$ iff for any invertible matrix $B\in\SS^m$, $BAB^\top\succeq 0$;
		\item[(iv)]  if $A\in\SS_+^m$ 
	and %$m-k+1$ is the largest index $j\in[m]$ such that $\lambda_{j}(A)=0$ 
	$\rank(A)=k$, then there exists a (unique) \add{invertible matrix $U$} such that, upto a permutation of the rows and columns of $A$, we can write  
    	\begin{align*}
		UAU^\top&=\bar\bI_k:= \left[\begin{array}{c|c}
		\bI_{k}&0\\
		\hline
		0&0
		\end{array}
		\right].
		\end{align*}
	\end{itemize}
		
\end{fact}

By Fact~\ref{f2}(i), if $\bx$ satisfies \raf{sdp}, then $\bI_{m}\bullet f(\bx)=\sum_{j=1}^n\bI_{m}\bullet A^{j}x_j\le\bI_{m}\bullet T$, which in turn implies that $x_j\le c_{j}':=\frac{\Tr(T)}{\Tr(A^{j})}$. Thus we may restrict the set of solutions for \raf{sdp} to the integer box $\cC:=\{\bx\in\RR^n|~0\le\bx\le \bc\}$, where $c_j:=\lfloor c_j'\rfloor$. Suppose that $\rank(T)=d$. By Fact~\ref{f2}(iv), we can write  $UTU^\top=\bar\bI_{d}$, for an \add{invertible} matrix $U$. Fact~\ref{f2}(iii) then implies that we can left-multiply by \add{$U^{-1}$} and right-multiply by \add{$U^{-\top}$} both sides of the $i$th inequality in~\raf{sdp} without changing the set of feasible solutions.  In other words, after possibly permuting the rows and columns of the matrices $A^{j}$, we can write~\raf{sdp} as follows:
	\begin{alignat}{3}
\label{sdp2}
\quad & \displaystyle \sum_{j=1}^nB^{j}x_j\preceq ~\bar \bI_{d},\\
\qquad &\bx\in\cC,\nonumber
\end{alignat}
where $B^{j}:=\add{U^{-1}}A^{j}\add{U^{-\top}}$. We further note by \add{Fact~\ref{f2}(ii)} that, if $b^{j}_{kk}>0$ for some $k>d$, then any feasible solution $\bx$ to~\raf{sdp} must have $x_j=0$. Let \add{$N:=\{j\in[n]~|~ b^{j}_{kk}=0 \text{ for all }k>d\}$}. Then, \add{Fact~\ref{f2}(ii)} also implies that, for all $j\in N$, $B^{j}$ can be written as:
	\begin{align}\label{decomp2}
B^{j}&:= \left[\begin{array}{c|c}
C^{j}&0\\
\hline
0&0
\end{array}
\right],
\end{align}
where $C^{j}\in\SS_+^{d}$.  Hence, we may consider the following inequality, equivalent to~\raf{sdp2}:
	\begin{alignat}{3}
\label{sdp3}
\quad & \displaystyle \sum_{j\in N}C^{j}x_j\preceq ~\bI_{d},\\
\qquad &\bx\in\cC.\nonumber
\end{alignat}
Let $\cF$ and $\cG$ be the sets of maximal feasible vectors for~\raf{sdp} and~\raf{sdp3}, respectively. Then, $|\cF|=|\cG|$ (as $\cF=\{(\bx,\bzero^{[n]\setminus N}):~\bx\in\cG\}$ where $\bzero^{[n]\setminus N}$ denotes a vector of zeros in positions $i\in [n]\setminus N$), while $|\cI(\cF)|\le|\cI(\cG)|+n-|N|$ (as $\cI(\cF)=\{(\bx,\bzero^{[n]\setminus N}):~x\in\cI(\cG)\}\cup\{\add{\b1^i}:~i\in[N]\setminus[n]\}$). 

\medskip

We will use the following fact, which is a generalization of Fact~\ref{f1}.
\begin{fact}[see, e.g., \cite{M02}]\label{f3}
	Let $p_1,\ldots,p_m~:~\RR^d\to\RR$ be real polynomials of maximum degree $D$, and denote by $Z_i:=\{\bx\in\RR^d~|~p_i(\bx)=0\}$ the  zero set of $p_i$. Then the number of cells (and, in fact, all the faces) in the arrangement of the surfaces $Z_1,\ldots,Z_m$ is at most $\Psi_{d,D}(m):=2(2D)^d\sum_{i=0}^d 2^i\binom{4m+1}{i}$, which is bounded by $\left(\frac{50Dm}{d}\right)^d$, for \add{$m\ge d\ge 2$}.  
\end{fact}	
\begin{theorem}\label{t5}
	Let $\cF$ be the set of maximal feasible vectors for~\raf{sdp}. Then for any non-empty subset $\cY\subseteq\cF$, it holds that 
	\begin{align}\label{sdp-bdd}
		|\cI(\cY)\cap \cI(\cF)|\leq \Psi_{d,2}(n(n-1)/2)n|\cY|=O(n)^{2d+1}|\cY|,
	\end{align}
	where $d:=\rank(T)$. %Consequently, if  $d=O(1)$, then all maximal feasible vectors for~\raf{sdp} can be enumerated in quasi-polynomial time.
\end{theorem}
\begin{proof}
		By the argument preceding the theorem, we may consider the equivalent inequality~\raf{sdp3}. Indeed, if we show the bound $|\cI(\cG')\cap \cI(\cG)|\leq \sum_{i=1}^r\Psi_{d}(|N|(|N|-1)/2)|N|\cdot|\cG'|$ for any $\cG'\subseteq\cG$, we get, for $\cY=\{(\bx,\bzero^{[n]\setminus N}):~\bx\in\cG'\}$, 
		\begin{align*}
		|\cI(\cY)\cap \cI(\cF)|\leq \Psi_{d,2}(|N|(|N|-1)/2)|N|\cdot|\cG'|+n-|N|\le \Psi_{d,2}(n(n-1)/2)n|\cY|. 
		\end{align*}
		Thus, for simplicity we will consider~\raf{sdp} and assume w.l.o.g. in the following that $N=[n]$, $T=\bI_{d}$ (and hence, $d=m$). \add{W.l.o.g. we also assume that $d\ge 2$.} 
		
	To show~\raf{sdp-bdd}, we proceed in a way similar to the proof of Theorem~\ref{t4}.  Denoting by $\add{\BB^d}:=\{\bx\in\RR^d:~\|\bx\|\le 1\}$ the $d$-dimensional unit ball centered at the origin, we can rewrite~\raf{sdp} in the following equivalent form:
	\begin{alignat}{3}
		\label{sdp4}
		\quad & \displaystyle \add{f_{\bu}(\bx):= \sum_{j=1}^n(A^{j}}\bullet \bu\bu^\top) x_j\leq ~1,\quad\text{ for }\bu\in\add{\BB^{d}}.
	\end{alignat} 
	For $\bu\in\add{\BB^{d}}$, let $\add{\bw^{\bu}}\in\RR^n_+$ be the vector \add{whose} $j$-th component is $w_j:=A^{j}\bullet \bu\bu^\top$, and $\cF_{\bu}$ be the set of maximal feasible solutions for the inequality $f_{\bu}(\bx)=(\bw^{\bu})^\top\bx\le 1$. \add{Then, $\cI(\cY) \cap  \cI(\cF)=\bigcup_{\bu\in \BB^d}^r\left(\cI(\cY) \cap  \cI(\cF_\bu)\right)$.} 
	\add{By Lemma~\ref{l2}}, it is enough to bound the number of of distinct permutations defined by the set of weights $\{\bw^{\bu}~|~\bu\in\BB^{d}(0,1)\}$:
	\begin{align}\label{e14}
	|\{\sigma_{\bw^{\bu}}~|~\bu\in\BB^{d}(0,1)\}|\le \Psi_{d,2}(n(n-1)/2).
	\end{align}
	Consider the system of inequalities $w_j\le w_{j'}$ for \add{distinct} $j,j'\in[n]$ (considering $\bu$ as a variable in $\RR^{d}$):
	\begin{align}\label{e15}
	\left(A_{j}-A_{j'}\right)\bullet \bu\bu^\top \le 0,\quad\text{ for }j\ne j'\in[n].
	\end{align}
	The inequality-defining polynomials in~\raf{e15} form an arrangement satisfying the conditions in Fact~\ref{f3} with $D:=2$ and \add{$m:=n(n-1)/2$}, and hence partitions $\RR^{d}$ into $\Psi_{d,2}(n(n-1)/2)$ cells. The theorem follows.
\end{proof}

\section{Some Open Questions}
We conclude with some open questions that naturally arise from the preceding work: 
\begin{itemize}
	\item[(O1)] For a polynomial inequity~\raf{mon-s}, where the function  $f$ is of the form~\raf{poly}, can we show a dual bound of the form $q(\cL,|\cY|)=\poly(n,\add{|\cH|},|\cY|,d)$, independent of $t$, $R$ and polynomial in  $d:=\max_{H,\add{j}} d_{H,j}$ (in comparison to the bound in~\raf{bd2--})?
	 \item[(O2)]  For a {\it single} linear inequality of the form~\raf{mon-s}, it is known that all maximal feasible solutions can be enumerated in polynomial time~\cite{BEGKM02-SICOMP,C87,PS94}. If $f$ is a polynomial of constant \add{number of variables per term}, then (B3) implies that all maximal feasible solutions can be enumerated in quasi-polynomial time via a dual-boundedness argument. It remains open whether a {\it polynomial} time enumeration  algorithm exists for a single polynomial inequality with \add{fixed number of variables per term, or at least, with fixed degree}.
\item[(O3)]  For an SOC inequity of the form~\raf{socp}, can one show a dual bound of the form $q(\cL,|\cY|)=\poly(n,d,|\cY|)$, as opposed to the bound in Theorem~\ref{t4} which depends exponentially on $d$? 
\item[(O4)]  Is there a  polynomial time algorithm for enumerating all maximal feasible solutions for  a {\it single} SOC inequity~\raf{socp}, when the number of rows $d$ is fixed (in comparison to a quasi-polynomial time algorithm that \add{follows from Theorem~\ref{t4}})?

\item[(O5)] Similar questions as (O3) and (O4) arise for a PSD inequality of the form~\raf{sdp}, considering the rank $d$ of the matrix $T$ as a parameter that can be either fixed or a part of the input. 
\end{itemize} 

\newcommand{\etalchar}[1]{$^{#1}$}

%\bibliographystyle{alpha}
%\bibliography{dualn}

\begin{thebibliography}{BGKM04}
	
	\bibitem[AW02]{AW02}
	R.~Ahlswede and A.~Winter.
	\newblock Strong converse for identification via quantum channels.
	\newblock {\em IEEE Transactions on Information Theory}, 48(3):569--579, 2002.
	
	\bibitem[BEG{\etalchar{+}}02]{BEGKM02-SICOMP}
	E.~Boros, K.~Elbassioni, V.~Gurvich, L.~Khachiyan, and K.~Makino.
	\newblock Dual-bounded generating problems: All minimal integer solutions for a
	monotone system of linear inequalities.
	\newblock {\em SIAM Journal on Computing}, 31(5):1624--1643, 2002.
	
	\bibitem[BEGK02]{BEGK02-MFCS}
	E.~Boros, K.~Elbassioni, V.~Gurvich, and L.~Khachiyan.
	\newblock Matroid intersections, polymatroid inequalities, and related
	problems.
	\newblock In {\em Mathematical Foundations of Computer Science 2002, 27th
		International Symposium, {MFCS} 2002, Proceedings}, volume 2420 of {\em
		Lecture Notes in Computer Science}, pages 143--154. Springer, 2002.
	
	\bibitem[BEGK03]{BEGK-DAM03}
	E.~Boros, K.~Elbassioni, V.~Gurvich, and L.~Khachiyan.
	\newblock An inequality for polymatroid functions and its applications.
	\newblock {\em Discrete Applied Mathematics}, 131:255--281, 2003.
	
	\bibitem[BGKM00]{BGKM00}
	E.~Boros, V.~Gurvich, L.~Khachiyan, and K.~Makino.
	\newblock Dual-bounded generating problems: Partial and multiple transversals
	of a hypergraph.
	\newblock {\em {SIAM} J. Comput.}, 30(6):2036--2050, 2000.
	
	\bibitem[BGKM01]{BGKM01}
	E.~Boros, V.~Gurvich, L.~Khachiyan, and K.~Makino.
	\newblock Dual-bounded generating problems: Partial and multiple transversals
	of a hypergraph.
	\newblock {\em SIAM Journal on Computing}, 30(6):2036--2050, 2001.
	
	\bibitem[BGKM04]{BGKM04}
	E.~Boros, V.~Gurvich, L.~Khachiyan, and K.~Makino.
	\newblock Dual-bounded generating problems: weighted transversals of a
	hypergraph.
	\newblock {\em Discret. Appl. Math.}, 142(1-3):1--15, 2004.
	
	\bibitem[BI95]{BI95}
	J.~C. Bioch and T.~Ibaraki.
	\newblock Complexity of identification and dualization of positive boolean
	functions.
	\newblock {\em Information and Computation}, 123(1):50--63, 1995.
	
	\bibitem[Bud88]{BBMW88}
	L.~Budach.
	\newblock {\em Algebraic and Topological Properties of Finite Partially Ordered
		Sets}.
	\newblock Teubner-Texte zur Mathematik. B.G. Teubner, 1988.
	
	\bibitem[Cra87]{C87}
	Y.~Crama.
	\newblock Dualization of regular boolean functions.
	\newblock {\em Discrete Applied Mathematics}, 16(1):79--85, 1987.
	
	\bibitem[FK96]{FK96}
	M.~L. Fredman and L.~Khachiyan.
	\newblock On the complexity of dualization of monotone disjunctive normal
	forms.
	\newblock {\em Journal of Algorithms}, 21:618--628, 1996.
	
	\bibitem[GK99]{GK99}
	V.~Gurvich and L.~Khachiyan.
	\newblock On generating the irredundant conjunctive and disjunctive normal
	forms of monotone \text{Boolean} functions.
	\newblock {\em Discrete Applied Mathematics}, 96-97(1):363--373, 1999.
	
	\bibitem[HJ90]{HJ90}
	R.~A. Horn and C.~R. Johnson.
	\newblock {\em Matrix analysis}.
	\newblock Cambridge University Press, 1990.
	
	\bibitem[KBE{\etalchar{+}}07]{KBEGM07}
	L.~Khachiyan, E.~Boros, K.~Elbassioni, V.~Gurvich, and K.~Makino.
	\newblock Dual-bounded generating problems: Efficient and inefficient points
	for discrete probability distributions and sparse boxes for multidimensional
	data.
	\newblock {\em Theor. Comput. Sci.}, 379(3):361--376, 2007.
	
	\bibitem[Lov83]{L83}
	L.~Lovasz.
	\newblock Submodular functions and convexity.
	\newblock In M.~Grotschel A.~Bachem and B.~Korte, editors, {\em Mathematical
		Programming: The State of the Art}, pages 235--257, New York, 1983.
	Springer-Verlag.
	
	\bibitem[Mat02]{M02}
	J.~Matousek.
	\newblock {\em Lectures on Discrete Geometry}.
	\newblock Springer-Verlag, Berlin, Heidelberg, 2002.
	
	\bibitem[PS94]{PS94}
	U.~N. Peled and B.~Simeone.
	\newblock An $o(nm)$-time algorithm for computing the dual of a regular boolean
	function.
	\newblock {\em Discrete Applied Mathematics}, 49(1-3):309--323, 1994.
	
	\bibitem[WX08]{WX08}
	A.~Wigderson and D.~Xiao.
	\newblock Derandomizing the ahlswede-winter matrix-valued chernoff bound using
	pessimistic estimators, and applications.
	\newblock {\em Theory Comput.}, 4(1):53--76, 2008.
	
	\bibitem[Y{\"{u}}c02]{Y02}
	{\"{U}}.~Y{\"{u}}ceer.
	\newblock Discrete convexity: convexity for functions defined on discrete
	spaces.
	\newblock {\em Discret. Appl. Math.}, 119(3):297--304, 2002.
	
\end{thebibliography}
\appendix
\section{Omitted Proofs}
\begin{proof}[Proof of Proposition~\ref{p1}]
	Suppose that $f$ is supermodular. Consider $j\in[k]$, $z\in\cC_j\setminus\{c_j\}$, and
	$\bx',\bx''\in\cC_1\times\ldots\times\cC_{j-1}\times\{z\}\times\cC_{j+1}\times\ldots\times\cC_k$ such that $\bx'\le\bx''$. To show that $\partial_f(\bx',j,z)\le\partial_f(\bx'',j,z)$,  we take $\bx:=\bx'\vee \bone^j$ and $\by:=\bx''$ in~\raf{sup} to get
	\begin{align*}
		f(\bx''\vee \bone^j)+f(\bx')&=f((\bx'\vee \bone^j)\vee\bx'')+f((\bx'\vee \bone^j)\wedge\bx'')\\
		&\ge f(\bx'\vee \bone^j)+f(\bx''),
	\end{align*}
	giving the desired inequality. On the other hand, suppose that $\partial_f(\bx,j,z)$ is monotone in $\bx\in\cC_1\times\ldots\times\cC_{j-1}\times\{z\}\times\cC_{j+1}\times\ldots\times\cC_k$, for any  $j\in[n]$ and $z\in\cC_j\setminus\{c_j\}$. Consider $\bx,\by\in\cC$. We need to show that~\raf{sup} holds. Let $S(\bx,\by):=\{j\in[k]:~x_j>y_j\}$. The proof is by induction on the size of $S(\bx,\by)$. If $S(\bx,\by)=\emptyset$ (meaning that $\bx\le\by$) then~\raf{sup} holds as an equality and there is nothing to prove. Otherwise, taking any $j\in S(\bx,\by)$ and using montonicity of $\partial_f(\cdot,j,z)$ for $z:=(\bx\wedge\by)_j+\add{\ell-1}=y_j+\add{\ell-1}$ and $\ell\in\{1,\ldots,x_j-y_j\}$, we obtain
	\begin{align}\label{ineqs}
		f(\by+\ell\bone^j)-f(\by+(\ell-1)\bone^j)\ge f(\bx\wedge\by+\ell\bone^j)-f(\bx\wedge\by+(\ell-1)\bone^j).
	\end{align}
	Summing~\raf{ineqs} over all $\ell\in\{1,\ldots,x_j-y_j\}$, we get 
	\begin{align}\label{e2}
		f(\by')-f(\by)\ge f(\bx\wedge\by+(x_j-y_j)\bone^j)-f(\bx\wedge\by).
	\end{align}
	where $\by':=\by+(x_j-y_j)\bone^j$.
	As $|S(\bx,\by')|<\add{|}S(\bx,\by)|$, we get by induction that
	\begin{align}\label{e3}
		f(\bx\vee\by')+f(\bx\wedge\by')\ge f(\bx)+f(\by').
	\end{align}
	Summing~\raf{e2} and~\raf{e3} and noting that $\bx\wedge\by+(x_j-y_j)\bone^j=\bx\wedge\by'$ and $\bx\vee\by'=\bx\vee\by$ \add{yield} the claim.
\end{proof}

\begin{proof}[Proof of Lemma~\ref{int-lem}(ii)]
	The proof is by induction on \add{$|U|\ge 2$} with the base case, \add{$|U|=2$, being easy to verify}. 
	%We consider a number of cases.
	
	%\smallskip
	
	%\noindent {\it Case 1.}	 $|\cS|=2$. Suppose $\cS=\{S_1,S_2\}$ for two distinct sets $S_1,S_2\subseteq U$. Then, there is a $T\in\cT$ such that $S_1\cap S_2\subseteq T$, and hence, $w(U)\ge w(S_1\cup S_2)=w(S_1)+w(S_2)-w(S_1\cap S_2)\ge w(S_1)+w(S_2)-w(T)\ge 2t_2-t_1.$  It follows that 
	%	$$
	%	|\cS|=2=\frac{(2t_2-t_1)-t_1}{t_2-t_1}\le\frac{w(U)-t_1}{t_2-t_1}|\cT|.
	%	$$
	%	Thus we may assume in the following that $|\cS|\ge 3$. 
	For $u\in U$, let $\cS(u):=\{S\in\cS:~u\in S\}$ and $\cT(u):=\{T\in\cT:~u\in T\}$. 
	%
	%	\smallskip
	%	
	%\noindent {\it Case 2.}	 $|\cS(u)|=1$ for some $u\in U$ and $w(u)\ge t_2-t_1$. Letting $\cS':=\{S\in\cS:~u\not\in S\}$, $\cT':=\{T\setminus\{u\}:~T\in\cT\}$, $U'=U\setminus\{u\}$, we would have that the families $\cS'$ and $\cT'$ satisfy the preconditions of the lemma with respect to the weight function $w:U'\to\RR_+$ and thresholds $t_1$ and $t_2$. Thus, we get by induction (as $|\cS|-1\ge 2$) that 
	%\begin{align*}
	%|\cS|-1=|\cS'|\le\frac{w(U')-t_1}{t_2-t_1}|\cT'|=\frac{w(U)-w(u)-t_1}{t_2-t_1}|\cT|\le\frac{w(U)-t_1}{t_2-t_1}|\cT|-\frac{w(u)}{t_2-t_1},
	%\end{align*}
	%which, using $w(u)\ge t_2-t_1$, gives $|\cS|\le\frac{w(U)-t_1}{t_2-t_1}|\cT|$. 
	%
	%	\smallskip
	%	
	%\noindent {\it Case 3.}	 $|\cS(u)|=1$ for some $u\in U$ and $w(u)=0$. Then letting $\cS':=\{S\setminus\{u\}:~S\in\cS\}$, $\cT':=\{T\setminus\{u\}:~T\in\cT\}$, $U'=U\setminus\{u\}$, we would have that the families $\cS'$ and $\cT'$ satisfy the preconditions of the lemma with respect to the weight function $w:U'\to\RR_+$ and thresholds $t_1<t_2':=t_2-w(u)$. Thus, we get by induction that 
	%\begin{align}\label{e4}
	%|\cS|=|\cS'|\le\frac{w(U')-t_1}{t_2'-t_1}|\cT'|=\frac{w(U)-w(u)-t_1}{t_2-w(u)-t_1}|\cT|=\frac{w(U)-t_1}{t_2-t_1}|\cT|,
	%\end{align}
	%as $w(u)=0$. 
	%
	%\smallskip
	%
	%\noindent {\it Case 4.} 
	Let $U_1:=\{u\in U~:~|\cS(u)|\le 1\}$ and $U_2=U\setminus U_1$. We may assume without loss of generality that $|\cT(u)|=0$ for all $u\in U_1$. \add{If $|U_2|=0$ then $\cS$ forms a partition on a subset of $U_1$, and $\cT$ contains at least one set (e.g. $\cT=\{\emptyset\}$). Then $w(U)=w(U_1)\ge \sum_{S\in\cS}|S|\ge t_2|\cS|$, and $\frac{w(U)-t_1}{t_2-t_1}|\cT|\ge |\cS|$.  Let us assume therefore that $U_2\neq\emptyset$. } 
	
	For any $u\in  U_2$, letting $U'(u):=U\setminus\{u\}$, $\cS'(u):=\{S\setminus\{u\}:~S\in\cS(u)\}$ and $\cT'(u):=\{T\setminus\{u\}:~T\in\cT(u)\}$, the sets $\cS'(u)$ and $\cT'(u)$ satisfy the preconditions of the lemma with respect to the weight function $w:U'(u)\to\RR_+$ and thresholds $t_1':=t_1-w(u)$ and $t_2':=t_2-w(u)$. Thus, we get by induction \add{(as $|U'(u)|\ge 2$)}, that 
	\begin{align}\label{e5}
		|\cS(u)|=|\cS'(u)|\le\frac{w(U'(u))-t_1'}{t_2'-t_1'}|\cT'(u)|=\frac{w(U)-t_1}{t_2-t_1}|\cT(u)|.
	\end{align}
	Let $\alpha=\max_{T\in\cT}w(T)$. Then $w(U_1)\le w(U)-\alpha$.
	Multiplying both sides of~\raf{e5} by $w(u)$ and summing up the resulting inequalities over $u\in U_2$, we get by  the threshold separability of $\cS$ and $\cT$,
	\begin{align}\label{e6}
		t_2|\cS|\le\sum_{S\in\cS}w(S)&\le w(U_1)+\sum_{u\in U_2}w(u)|\cS(u)|\le w(U_1)+\sum_{u\in U_2}w(u)\frac{w(U)-t_1}{t_2-t_1}|\cT(u)|\nonumber\\ &=w(U_1)+\frac{w(U)-t_1}{t_2-t_1}\sum_{T\in \cT}w(T)\le w(U_1)+\frac{w(U)-t_1}{t_2-t_1}\alpha|\cT|\nonumber\\
		&\le w(U)-\alpha+\frac{w(U)-t_1}{t_2-t_1}\alpha|\cT|,
	\end{align}
	where $\alpha\le t_1$. Note that the right hand side of \raf{e6} is monotone \add{non-decreasing} in $\alpha$ and hence is maximized at $\alpha=t_1$. It follows that
	\begin{align}\label{e7}
		|\cS|\le \frac{w(U)-t_1}{t_2}+\frac{w(U)-t_1}{t_2-t_1}\cdot\frac{t_1}{t_2}|\cT|.
	\end{align}
	Using \add{$t_1<w(U)$} in \raf{e7}, we obtain the stated claim.
\end{proof}

\begin{proof} [Proof of Proposition~\ref{ch1:unbdd}]
	The result follows from the following reduction from the so-called 
	{\em relay cuts} enumeration problem in a relay circuit with two terminals. 
	Let $G=(V,E)$ be a graph with vertex set $V$ and edge set $E$, and two distinguished vertices $s,\add{g}\in V$. 
	To each edge in $e\in E$, is assigned a relay $j(e)\in[n]$ from a given set
	of relays $[n]$ (two or more distinct edges may be assigned identical relays).  
	Let \add{$\cG$} be the family of all minimal {\em $s$-\add{$g$} relay cuts}, i.e., minimal subsets of relays
	that disconnect $s$ and \add{$g$}. It is known that the problem of incrementally generating \add{$\cG$} is NP-hard, see \cite{GK99}. We define a polynomial $f: \{0,1\}\to\add{\ZZ_+^n}$ as follows. Let $\cP$ be the set of \add{walks} between $s$ and $g$ of length $|V|$ in $G$. We associate a variable $x_j$ to each relay $j\in[n]$, and for $\bx\in\{0,1\}^n$, we let 
	\begin{align}\label{relay-f}
		f(\bx)=\sum_{P\in\cP}\prod_{e\in P}x_{j(e)}.
	\end{align}
	Given $\bx\in\{0,1\}^n$, we can compute $f(\bx)$ 
	in polynomial time (this requires only computing the \add{$|V|$-{th}} power of the adjacency
	matrix of the graph $G(\bx)$ obtained from $G$ by deleting all edges \add{$e\in E$} with $x_{j(e)}=0$. \add{In particular}, checking \add{if $f(\bx)\le 0$} is equivalent to checking if there is \add{no} \add{$s$-$g$} path in $G(\bx)$).
	This gives a polynomial time evaluation oracle for $f$. Finally, note
	that minimal $s$-\add{$g$} relay
	cuts are exactly the complements of the maximal feasible solutions $\cF$ of the polynomial
	inequality $f(\bx)\le 0$. 
\end{proof}

\begin{proof}[Proof of Lemma~\ref{l2}]
	For a vector $\bx\in\cC\setminus \{\bzero\}$ \add{and $i\in U$}, let us denote by $j^{\bx}_i$ the index
	of the last component, in the order given by $\sigma_i$, which is larger than $0$, i.e.,
	$j_i^\bx=\max\{j\in[n]~|~x_{\sigma_i(j)}>0\}$. For a vector $\by\in\cC$\add{, index $j\in[n]$} and a permutation $\sigma$, denote by $\by^{\sigma, j}$ the vector $\by'$ with components:
	\begin{equation}\label{eb3}
		y'_{\sigma(j')}=\left\{
		\begin{array}{ll}
			y_{\sigma(j')}&\text{ for } j'<j,\\
			y_{\sigma_i(j')}+1&\text{ for } j'=j,\\
			0&\text{ otherwise}.
		\end{array}
		\right.
	\end{equation}
	Let $\cF_i$ be the set of maximal feasible solutions of the inequality $f_i(\bx)\le t_i$. 
	We claim that for every $\bx\in \cX:=\cI(\cY)\cap \cI(\cF)$ there exists  an $i\in[r]$ and a
	$\by \in\cY$ such that $\bx=\by^{\sigma_i,j^\bx_i}$.
	To see this claim, let us consider $\bx\in \cI(\cY)\cap \cI(\cF)$ and observe that $\bx \ne \bzero$ because  $\bx \in \cI(\cF)$ and  $\cF\neq \emptyset$. As $\cI(\cY) \cap  \cI(\cF)=\bigcup_{i=1}^r\left(\cI(\cY) \cap  \cI(\cF_i)\right)$, there exists an $i\in[r]$ such that $\bx\in \add{\cX_i}:=\cI(\cY)\cap \cI(\cF_i)$. Let $j:=j_i^{\bx}$. Then, as $\bx\in\cI(\cY)$, there exists a $\by\in \cY$ such that
	$\by\geq \bx-\add{\b1^{\sigma_i(j)}}$. For any $j'<j$, we  must have
	$x_{\sigma_i(j')}=y_{\sigma_i(j')}$, since if $x_{\sigma_i(j')}<y_{\sigma_i(j')}$ for some $j'<j$, then
	$f_i(\by)\geq f_i(\bx+\b1^{\sigma_i(j')}-\b1^{\sigma_i(j)})\ge \add{f_i}(\bx)>t_i$ would follow by the 2-monotonicity of $\add{f_i}$, and yielding 
	a contradiction with $f_i(\by)\le t_i$ (which follows from $y\in\cF$). Finally,  the
	definition of $j=j^\bx_i$ implies that $x_{\sigma_i(j')}=0$ for all $j'>j$. Hence,
	our claim and the equality \raf{eb3} follow.
	
	The  above claim implies that
	\[
	\cX\subseteq \{ \by^{\sigma_i,j}~|~\by\in\cY,~i\in\add{U},~j\in[n],~ y_{{\sigma_i(j)}}<c_{\sigma_i(j)}\},
	\]
	and hence \raf{2mon-dbdd} and thus the lemma follow.
\end{proof}
\end{document}